\documentclass[11pt]{article}

\usepackage[left=1in, top=1in, right=1in, bottom=1in]{geometry}
\usepackage{amsmath}
\usepackage{amssymb}
\usepackage{amsthm}
\usepackage{paralist}
\usepackage{algorithmicx}
\usepackage[noend]{algpseudocode}
\usepackage[linesnumbered,lined,boxed,commentsnumbered]{algorithm2e}
\usepackage{color}
\usepackage[square,numbers]{natbib}
\usepackage{aliascnt}
\usepackage{tcolorbox}

\usepackage{tikz}
\usepackage{pgfplots}

\usepackage{mathtools}
\usepackage{bm}
\usepackage{todonotes}

\newcommand*\samethanks[1][\value{footnote}]{\footnotemark[#1]}

\newcommand{\AutoAdjust}[3]{\mathchoice{ \left #1 #2  \right #3}{#1 #2 #3}{#1 #2 #3}{#1 #2 #3} }
\newcommand{\Xcomment}[1]{{}}

\newcommand{\InBrackets}[1]{\AutoAdjust{[}{#1}{]}}
\newcommand{\Ex}[2][]{\operatorname{\mathbf E}_{#1}\InBrackets{#2}}
\newcommand{\Prx}[2][]{\operatorname{\mathbf{Pr}}_{#1}\InBrackets{#2}}
\def\prob{\Prx}
\def\expect{\Ex}

\newcommand{\dd}{\mathrm{d}}  

\newcommand{\eps}{\epsilon}

\DeclarePairedDelimiter\abs{\lvert}{\rvert}
\DeclareMathOperator{\Ber}{Ber}
\DeclareMathOperator{\Binom}{Binom}

\newcommand{\R}{\mathbb{R}}
\newcommand{\cA}{\mathcal{A}}

\newcommand{\noaccents}[1]{#1}
\newcommand{\newagentvar}[3][\noaccents]{%
\expandafter\newcommand\expandafter{\csname #2\endcsname}{#1{#3}}%
\expandafter\newcommand\expandafter{\csname #2s\endcsname}{#1{\boldsymbol{#3}}}%
\expandafter\newcommand\expandafter{\csname #2smi\endcsname}[1][i]{#1{\boldsymbol{#3}}_{-##1}}%
\expandafter\newcommand\expandafter{\csname #2i\endcsname}[1][i]{#1{#3}_{##1}}%
\expandafter\newcommand\expandafter{\csname #2ith\endcsname}[1][i]{#1{#3}_{(##1)}}%
}

\newagentvar{type}{t}
\newagentvar{alloc}{x}
\newagentvar{pay}{p}
\newagentvar{typespace}{T}
\newagentvar{sigspace}{S}
\newagentvar{val}{v}
\newagentvar{util}{u}
\newcommand{\bval}{\vals}
\newcommand{\balloc}{\mathbf{\alloc}}
\newcommand{\bpay}{\mathbf{\pay}}

\newagentvar{virt}{\varphi}

\DeclareMathOperator*{\argmax}{arg\,max}

\DeclareMathOperator{\Rev}{Rev}
\DeclareMathOperator{\OPT}{OPT}

\DeclareMathOperator{\VCG}{VCG}

\newagentvar[\tilde]{exanteq}{q}

\newagentvar{sample}{s}

\DeclareMathOperator{\LA}{LA}

\newcommand{\sechigh}{\vali[\mathrm{sh}]}
\newcommand{\res}{r}
\newcommand{\monres}{\res^*}

\newcommand{\unknown}{s}

\newtheorem{theorem}{Theorem}[section]%
\newtheorem{lemma}[theorem]{Lemma}%
\newtheorem{fact}[theorem]{Fact}
\newtheorem{claim}[theorem]{Claim}
\newtheorem{definition}[theorem]{Definition}
\newtheorem{example}[theorem]{Example}
\newtheorem{remark}[theorem]{Remark}

\newcommand{\ClaimName}[1]{\label{claim:#1}}
\newcommand{\Claim}[1]{Claim~\ref{claim:#1}}
\newcommand{\EquationName}[1]{\label{eqn:#1}}
\newcommand{\Equation}[1]{Equation~\eqref{eqn:#1}}
\newcommand{\FactName}[1]{\label{fact:#1}}
\newcommand{\Fact}[1]{Fact~\ref{fact:#1}}
\newcommand{\LemmaName}[1]{\label{lem:#1}}
\newcommand{\Lemma}[1]{Lemma~\ref{lem:#1}}
\newcommand{\TheoremName}[1]{\label{thm:#1}}
\newcommand{\Theorem}[1]{Theorem~\ref{thm:#1}}

\newcommand{\Subsection}[1]{Subsection~\ref{subsec:#1}}
\newcommand{\SubsectionName}[1]{\label{subsec:#1}}
\newcommand{\Section}[1]{Section~\ref{sec:#1}}

\title{The Vickrey Auction with a Single Duplicate Bidder Approximates the Optimal Revenue}
\date{\today}

\author{
	Hu Fu\thanks{Department of Computer Science, University of British Columbia. \tt{\{hufu, cvliaw, srand\}@cs.ubc.ca}}
	\and
	Christopher Liaw\samethanks	
	\and
    Sikander Randhawa\samethanks
}

\begin{document}

\maketitle


Bulow and Klemperer's well-known result states that, in a single-item auction where the $n$ bidders' values are independently and identically drawn from a regular distribution, the Vickrey auction with one additional bidder (a duplicate) extracts at least as much revenue as the optimal auction without the duplicate.  Hartline and Roughgarden, in their influential 2009 paper, removed the requirement that the distributions be identical, at the cost of allowing the Vickrey auction to recruit $n$ duplicates, one from each distribution, and relaxing its revenue advantage to a $2$-approximation.

In this work we restore Bulow and Klemperer's number of duplicates in Hartline and Roughgarden's more general setting with a worse approximation ratio.  We show that recruiting a duplicate from one of the distributions suffices for the Vickrey auction to $10$-approximate the optimal revenue.
We also show that in a $k$-items unit demand auction, recruiting $k$ duplicates suffices for the VCG auction to $O(1)$-approximate the optimal revenue.

As another result, we tighten the analysis for Hartline and Roughgarden's Vickrey auction with $n$ duplicates for the case with two bidders in the auction.  We show that in this case the Vickrey auction with two duplicates obtains at least $3/4$ of the optimal revenue.  This is tight by meeting a lower bound by Hartline and Roughgarden.  En route, we obtain a transparent analysis of their $2$-approximation for $n$~bidders, via a natural connection to Ronen's lookahead auction.

\newpage
\cleardoublepage
\setcounter{page}{1}

\section{Introduction}
\label{sec:intro}

Bulow and Klemperer's theorem \cite{BK96} is a fundamental result in the auction literature, drawing a connection between two basic auctions for selling a single item.  The second price auction, also known as the Vickrey auction \citep{Vic61}, lets the highest bidder win the item at the price of the second highest bid.  It always sells to the highest bidder, thereby maximizing the market's efficiency, requires no prior information on the bidders, and is easy to implement in practice.  On the other hand, revenue optimal auctions, those that generate the most revenue for the auctioneer, are usually more intricate.  They may sometimes not sell to anyone, or sell to a bidder who does not place the highest bid, and their rules in general depend on prior information on bidders' values \citep{M81}.  \citeauthor{BK96} showed that, in the case where bidders' values are drawn independently and identically from a distribution that satisfies a commonly assumed property known as regularity, the simple Vickrey auction with $n+1$ bidders extracts no less revenue than the revenue optimal auction with $n$ bidders, for any $n \geq 1$.  This fundamental result makes an elegant comparison between the revenue impact brought by enhanced competition and that attributed to more precise market information and tailored auction design.  It holds for a large family of distributions, as long as the bidders' values are i.i.d.\@ drawn.

\citet{HR09}, in their seminal work that started a fruitful line of research on the revenue of simple auctions (as exemplified by the Vickrey auction), extended \citeauthor{BK96}'s result, with approximation, when one removes the assumption of symmetry, i.e., when bidders' values are drawn independently but not identically from regular distributions.  
This extension is approximate for two changes it makes in the comparison between the two auctions. First, the Vickrey auction, instead of recruiting only one additional bidder, is now allowed to recruit $n$ additional bidders (called duplicates), one from each distribution. Second, the revenue of this Vickrey auction with $n$ duplicates, instead of being lower bounded by the optimal revenue without duplicates, is shown to be at least half as much.  
Bidders with different value distributions are often interpreted as coming from different populations or possessing different characteristics.  Under this interpretation, the Vickrey auction is shown to approximate the optimal revenue if one is allowed to recruit at least one additional bidder from each population.
\citeauthor{HR09} also gave an example with two bidders where the Vickrey auction, even with two duplicate bidders, one from each distribution, extracts only $3/4$ of the revenue of the optimal auction without duplicates.

%
The two changes made in \citeauthor{HR09}'s result leave two gaps between the behavior of auctions with i.i.d.\@ bidders and those with independent but non-identical ones.
The first gap is a necessary one: \citeauthor{HR09}'s example shows that, absent symmetry among bidders, a constant fraction of revenue has to be lost by using the simpler Vickrey auction, even when all bidders are duplicated. 
The other gap has been an open question: \emph{how many duplicate bidders} are needed for the Vickrey auction to give a constant approximation to the optimal revenue?  Whereas for the i.i.d.\@ case a single duplicate bidder suffices, it has remained unknown whether $n$ duplicates are necessary when the value distributions are not identical.

\subsection{Main Results}

\paragraph{Vickrey auction with a single duplicate.} In this work, we close this second gap.  We show that, in a single-item auction with $n$ bidders, whose values are drawn independently from possibly different regular distributions, the Vickrey auction with \emph{one} additional bidder with value drawn from \emph{one of the existing distributions} obtains at least $\tfrac 1 {10}$ of the optimal revenue extractable from the original bidders.\footnote{In fact, it extracts at least $\frac{1}{10}$ of the optimal \emph{ex ante} revenue from the original bidders.}  Therefore, to retain a constant fraction of the optimal revenue, the number of duplicate bidders needed in the non-i.i.d.\@ setting is in fact no more than in the i.i.d.\@ setting.

Our analysis parts ways from the very beginning with that by Hartline and Roughgarden, who use the Myersonian ``virtual surplus'' as an (exact) upper bound on the optimal revenue without duplicates.  Instead, we work with the looser bound given by the ex ante optimal revenue.  The looseness of this bound is compensated by the ease with which it allows connections to be drawn between the optimal revenue and values (or prices) in bidders' distributions, and all such quantities across bidders to be manipulated in a linear fashion.  With this, we identify a set of bidders that, in total, contribute a constant fraction of the ex ante optimal revenue by using high prices (so high as to be comparable to the optimal revenue).  Two observations are in order: (a) if a bidder accepts a high price with a constant probability, then duplicating this bidder alone suffices for the Vickrey auction to extract a constant fraction of the optimal revenue; and (b) if there is a high price such that the \emph{sum} of probabilities with which bidders take the price is a constant, then the VCG auction even \emph{without} duplicates is a constant approximation.  A technical lemma (Lemma~\ref{lem:single-technical}) shows that we must, in fact, be in one of these scenarios.  The lemma makes use of regularity and the set of bidders we identified in the previous step.  


\paragraph{Choosing the distribution to duplicate.}
Naturally, the distribution to be duplicated needs to be carefully chosen; for instance, if most bidders' values are constantly $0$ or negligible, duplicating an arbitrary or a random bidder helps little with the revenue.  
We show several forms of minimal information that the auctioneer could acquire in order to choose the distribution to duplicate: to guarantee a constant approximation achieved by the Vickrey auction with a duplicate, it suffices that the auctioneer knows (approximately) the value of each bidder at a specified quantile; that is,  the auctioneer only needs to know, for each bidder~$i$,  what is roughly the value~$\vali$ such that with probability~$q$ bidder~$i$ bids above~$\vali$, where $q$ is a fixed probability, say, $0.5$.
Note that the approximation ratio may depend on how well the auctioneer can query the value of a bidder at the specified quantile $q$ as well as the quantile itself.
It also suffices if the auctioneer has a single sample from each distribution albeit with a different constant approximation (see \Theorem{single-beta-exact}).
Our result may be seen therefore as describing a tradeoff between reducing the number of duplicate bidders and acquiring some minimal knowledge on the value distributions.  

We remark that we do not claim that the best use of such prior information is to find a duplicate and run the Vickrey auction.  In fact, with the said types of information, one may run other auctions for comparable revenue guarantees.
For example, with a single sample from each distribution, one may use the maximum of these samples as an anonymous reserve in the Vickrey auction
and guarantee a $4$-approximation \cite{HR09}.
We see the main contribution of this work as conceptually closing the long-standing gap in the understanding on the number of duplicates needed for the Vickrey auction to be approximately revenue optimal.  
The Vickrey auction is fundamental and ubiquitous, and, partly owing to this, Bulow and Klemperer's original result is a central link in the auction literature; our result extends it to non-i.i.d.\@ settings in a spirit close to the original.  
 The results on choosing the duplicate are natural consequences of our techniques used to prove the main result.

On the other hand, as we discuss at the end of the paper, the upper bounds we give are unlikely to be tight; for all that we know, it is possible that a tighter analysis could show the VCG auction with a duplicate to be a comparable or even better way to utilize the kind of partial information we discussed.  We leave this intriguing question to future work.

\paragraph{$k$ duplicates for $k$-items auctions with unit demands.}  We extend our result to $k$-items auctions with unit demands, where there are $k$ identical items for sale and each bidder needs only one of them. The generalization of the Vickrey auction here is the VCG auction, which sells to the highest $k$ bidders at the price of the $(k+1)$-st highest bid.  
A simple example shows that one needs more duplicates to generate \emph{any} revenue in the VCG auction: 
let there be one bidder whose value is~$1$ and all the other bidders have value~$0$.  Here even if we duplicate all bidders, the VCG auction still has revenue~$0$, whereas the optimal revenue is~$1$.
To address this, we need to be allowed to make \emph{multiple} duplicates of the same bidder.  The above example shows that $k$ duplicates are necessary for the VCG auction to give any approximation.
We show that this is tight: as long as the bidders' values are drawn independently from regular distributions, there exists a bidder such that one may add $k$ duplicates of her to the VCG auction and guarantee a constant fraction of the optimal revenue without duplicates.

\citet{HR09} addressed the issue in the example in a different way: the VCG auction may still add only one duplicate for each bidder, but each original bidder and her duplicate can win at most one item.  With this restriction, they show that the VCG auction with a duplicate for every bidder is again a $2$-approximation.  We extend our analysis to this scenario as well, and show that duplicating $k$ bidders suffices for the VCG auction to be a constant approximation, albeit with a smaller constant.
The following example shows the tightness of our result: let $k / 3$ bidders have value~$1$ and all other bidders have value~$0$, then one must duplicate $\Omega(k)$ bidders to guarantee a constant approximation.


\paragraph{Tighter analysis of the Vickrey auction with $n$ duplicates.}  

Since \citeauthor{HR09} first gave the upper bound of~$2$ and lower bound of $\tfrac 4 3$ on the approximation factor of the Vickrey auction with $n$ duplicates, it has remained an open question what this factor is in the worst case.  
We make progress towards resolving this by showing that for $n = 2$, $\tfrac 4 3$ is the tight bound.  We first give an alternate proof of \citeauthor{HR09}'s original $2$-approximation.  Our proof uses the geometry of the so-called revenue curves and lower bounds the revenue of the Vickrey auction with $n$ duplicates by that of Ronen's \emph{lookahead auction} \cite{R01}, which is known to $2$-approximate the optimal revenue.
We then further exploit properties of the Vickrey auction and also switch to a \emph{stronger} benchmark, the so-called \emph{ex ante} optimal revenue.  This allows us to identify a small family of distributions as scenarios; a thorough analysis of these reveals the tight ratio of $\tfrac 4 3$.

\paragraph{Related literature.}

From the large literature on ``simple versus optimal auctions'' that followed \citet{HR09}, we point out a few that are closer to our work.  
\citet{SS13} extended \citeauthor{HR09}'s result to distributions that are convex combinations of regular ones (and the duplicate bidders are drawn from each component regular distribution).
\citet{RTY12} showed a Bulow-Klemperer type result for matching environments. \citet{EFFTW17} and \citet{FFR18} studied, for auctions selling multiple \emph{heterogenous} items, the number of duplicate bidders needed for the VCG auction's revenue to approximate the optimal without duplicates.  The bidders' preferences are \emph{i.i.d.\@} drawn, and the non-trivial number of duplicates needed is a consequence of multi-dimensional preferences rather than asymmetry among bidders.

Our results on the minimal information needed to choose the bidder to duplicate also place the work in the literature on prior-independent mechanisms \citep[e.g.][]{AB18, DRY15, FILS15, CHMS13}.  In particular, we show in Theorem~\ref{thm:single-duplicate} that a \emph{single sample} in our setting also suffices to choose the duplicate.

\paragraph{Organization of the paper}

In Section~\ref{sec:prelim} we set up the model formally and give preliminaries.  In Section~\ref{sec:single} we show our result on duplicating few bidders, first for single-item auctions and then for $k$-items auctions with unit demands in \Section{k-item}.  In Section~\ref{sec:2approx} we show our tighter analysis for the auction with $n$ duplicated bidders.  We conclude with some discussion in Section~\ref{sec:conclusion}.

\section{Preliminaries}
\label{sec:prelim}
\paragraph{Single-item auctions.}
A single item is to be sold to $n$ bidders.  
Each bidder $i$ values the item at~$\vali$, which is private to~$i$ and independently drawn from a publicly known distribution~$F_i$.  
We use $F_i$ interchangably as both the distribution of bidder $i$'s value as well as the cumulative density function (cdf) of the distribution, i.e.~$F_i(w) = \prob{v_i \leq w}$.
We will also use the notation $\bval = (\val_1, \ldots, \val_n)$ for the vector containing the values of all bidders and $\bval_{-i} = (\val_1, \ldots, \val_{i-1}, \val_{i+1}, \ldots, \val_n)$ for the vector \emph{excluding}
bidder $i$'s value.

A selling mechanism consists of an allocation rule $\balloc(\bval) = (\alloc_1(\bval), \ldots, \alloc_n(\bval))$, and a payment rule $\bpay(\bval) = (\pay_1(\bval), \ldots, \pay_n(\bval))$,
where $\alloc_i(\bval) \in [0,1]$ indicates the probability that bidder~$i$ receives the item when the bid vector is~$\vals$,
 and $\pay_i(\bval) \in \R$ indicates the amount that bidder $i$ must pay.
The utility of bidder $i$ is quasi-linear, given by $\util_i = \alloc_i(\bval) \val_i - \pay_i(\bval)$.
Note that we must enforce the condition $\sum_i \alloc_i(\bval) \leq 1$ for all $\bval$ as there is only one item for sale.

A mechanism is said to be \emph{dominant strategy incentive compatible} (DSIC) if for all~$i$, any value profile~$\bval$, it holds that for any deviation $\val_i'$, we have
\begin{equation}
    \EquationName{DSIC}
    \alloc_i(\bval) \val_i - \pay_i(\bval) \geq \alloc_i(\val_i', \bval_{-i}) \val_i - \pay_i(\val_i', \bval_{-i}).
\end{equation}
A mechanism is said to be \emph{Bayesian incentive compatible} (BIC) if for any~$i$, any value~$\vali$ and any deviation $\val_i'$, we have
\begin{equation}
    \EquationName{BIC}
    \Ex[\valsmi]{\alloci(\vali, \valsmi) \vali - \payi(\vali, \valsmi)} \geq \Ex[\valsmi]{\alloci(\vali', \valsmi) \vali - \payi(\vali, \valsmi)}.
\end{equation}
A mechanism is said to be \emph{ex-post individually rational} (IR) if a bidder's utility is always nonnegative, i.e.~for all $i$ and $\bval$, we have
\begin{equation}
    \EquationName{IR}
    \alloc_i(\bval) \val_i - \pay_i(\bval) \geq 0.
\end{equation}

As will be clear in the summary of \citeauthor{M81}'s work below, for our purpose there is no need to make a distinction between DSIC and BIC.  Without loss of generality, we focus in this work on DSIC and ex-post IR mechanisms.  In such mechanisms, one may assume that bidders bid truthfully.  The \emph{revenue} of a mechanism is then $\Ex[\vals]{\sum_i \payi(\vals)}$.

\paragraph{Myerson's characterization of optimal mechanisms.}
\citet{M81} completely characterized the revenue of incentive compatible auctions:
\begin{lemma}[\citet{M81}]
    \LemmaName{Myerson}
    \begin{enumerate}
	    \item 
            The payment rule of any BIC mechanism is uniquely determined by its allocation rule up to a constant.  With this, a mechanism is BIC if and only if its allocation rule $\alloci(\vali, \valsmi)$ for each~$i$ is monotone nondecreasing in~$\vali$ for any $\valsmi$.  
	    \item For any BIC mechanism with allocation rule $\allocs$, its expected revenue, up to a constant, is given by $\Ex[\vals]{\sum_i \alloci(\vali) \virti(\vali)}$, where $\virti (\vali) \coloneqq \vali - \frac {1 - F_i(\vali)}{f_i(\vali)}$ is called the \emph{virtual value} of~$\vali$.
	    \item There exists a DSIC and ex-post IR mechanism that is revenue optimal amongst all \emph{BIC} and ex-post IR mechanisms.
    \end{enumerate}
\end{lemma}
Throughout the paper, we fix the constant so that the lowest type in the support has expected utility $0$.

A distribution is \emph{regular} if the virtual value $\virt(\val)$ is monotone nondecreasing in~$\val$.  Given the characterization, it is immediate that, for regular distributions, the revenue-optimal auction allocates to the bidder with the highest virtual non-negative value; if all virtual values are negative, the item is not sold.  We will refer to the revenue of Myerson's optimal auction as the \emph{optimal revenue}.

\paragraph{Revenue curves.} 
\citet{BR89} gave an influential reinterpretation of \citeauthor{M81}'s characterzation, and we will use heavily this interpretation.  Given a bidder whose value is drawn from distribution~$F$, for every take-it-or-leave-it price~$p$, she will buy with probability $1 - F(p)$, thereby generating revenue $p(1 - F(p))$.  Let $q(p) \coloneqq 1 - F(p)$ be the probability of selling at~$p$, and call it the \emph{quantile} of~$p$.  The plot of revenue $\Rev(q) \coloneqq q\cdot F^{-1}(1 - q)$ against~$q \in [0, 1]$ is called the \emph{revenue curve} of the distribution.  With a slight overloading of notation, for a quantile~$q$ we denote by $\val(q) \coloneqq F^{-1}(1 - q)$ its corresponding value.

It is without loss of generality to assume that $\Rev(0) = \Rev(1) = 0$.\footnote{See, e.g.\@ Appendix~A of \citet{FILS15} for an argument.}  The most important fact about the revenue is that the derivative of $\Rev(q)$ at quantile~$q$ is equal to the virtual value of~$\val(q)$.  In this paper, we do not appeal to this connection other than the curve's concavity for regular distributions.  

The highest point of a revenue curve gives the maximum revenue extractable from a bidder by posting a take-it-or-leave-it price.  This price is called the \emph{monopoly reserve} for the bidder.  For bidder~$i$ we denote this by $\monres_i$.  When there is a single bidder in an auction, the revenue yielded by posting $\monres$ is in fact the optimal revenue, by Myerson's characterization.

\paragraph{The Vickrey Auction.}  
In the Vickrey auction (a.k.a.\@ the second price auction), the item is allocated to the bidder with the highest value, who is then charged the bid of the second highest bidder.  \citet{Vic61} showed that this auction is DSIC and ex-post IR.

\paragraph{Auctions with duplicate bidders and Bulow-Klemperer type results.}
In general, the revenue of the Vickrey auction can be far from the optimal revenue.
\citet{BK96} showed that, when bidders' values are drawn i.i.d.\@ from a regular distribution, the second price auction, by recruiting one additional bidder (called a \emph{duplicate bidder}), could reverse the comparison:
\begin{theorem}[\citet{BK96}]
    \TheoremName{BK}
    Let $\OPT$ be the optimal revenue for an $n$ bidder, single-item auction where each bidder's value is drawn independently from an identical regular distribution.
    Then the expected revenue of the second price auction with $n+1$ independent bidders from the same distribution is at least $\OPT$.
\end{theorem}
\citet{HR09} extended this theorem in an approximate manner to settings where bidders' values are independently but not identically drawn:
\begin{theorem}[\citet{HR09}]
    \TheoremName{HR}
    Let $\OPT$ be the optimal revenue extractable in an auction with $n$ bidders where bidder~$i$'s value is drawn independently from a regular distribution $F_i$.  
    The expected revenue of the second price auction with $2n$ bidders, where the values of bidder~$i$ and $n + i$ are independently drawn from~$F_i$,  is at least $\tfrac 1 2 \OPT$.
\end{theorem}

\citeauthor{HR09} also gave an example with $n = 2$ in which the Vickrey auction with $2n$ bidders extracts a revenue that is $\tfrac 3 4 \OPT$.
We recall this example in Section~\ref{sec:2approx}.

\paragraph{Lookahead auction.}
The lookahead auction, first devised and analyzed by \citet{R01},  offers a technical tool to analyze auctions.
At a bid profile $\bval$, the lookahead auction first identifies the highest bidder, say $i^*$, then offers the item to $i^*$ at the optimal posted price for $i^*$ conditioned on $\valsmi$ and the fact $\val_{i^*} \geq \max_{j \neq i} \val_{j}$.  

\begin{theorem}[\citet{R01}]
	\TheoremName{lookahead}
    The lookahead auction is a DSIC, ex-post IR mechanism whose revenue 2-approximates the optimal revenue.
\end{theorem}

\paragraph{Ex ante optimal revenue.}
Almost all approximation results we present in this paper are in fact with respect to the \emph{ex ante optimal revenue}, a benchmark that is usually strictly stronger than the optimal revenue.  Given regular distributions with revenue curves $\Rev_1, \cdots, \Rev_n$, the ex ante optimal revenue is the value of the following program:
\begin{equation}
\EquationName{ExAnteSingle}
\begin{aligned}
    \max  & 
	\sum_{i=1}^n \Rev_i(q_i) \\
    \text{s.t. }& \sum_{i=1}^n q_i \leq 1.
\end{aligned}
\end{equation}
We denote by $\exanteqi[1], \ldots, \exanteqi[n]$ the optimal solution to this program.  We refer the reader to \citet{MDnA} for an exposition that the ex ante optimal revenue upper bounds the optimal revenue.

\paragraph{$k$-Items Unit-Demand Auctions.}
In a $k$-items unit-demand auction, there are $k$ identical items to sell, and each bidder wants at most one item.  Each bidder's value for an item is, as before, drawn independently.  All the notions defined above carry over to this setting, with minimal change as follows: the feasibility constraint on the allocation rule now becomes, for all~$\vals$, $\sum_i \alloci(\vals) \leq k$, and $0 \leq \alloci(\vals) \leq 1$, for all~$i$.  The generalization of the Vickrey auction in this setting is the VCG auction, which allocates the items to the $k$ highest bidders, and then charges each winner the $(k+1)$-st highest bid.  In the program that defines the ex ante optimal revenue, the constraint now becomes $\sum_i q_i \leq k$, and $q_i \leq 1$ for all~$i$.  Myerson's characterization (Lemma~\ref{lem:Myerson}) still holds, and the revenue-optimal auction allocates the items to at most $k$ bidders with the highest non-negative virtual values.

As we explained in the introduction, in order for VCG with duplicates to be approximately revenue optimal, one needs the further restriction that an original bidder and her duplicate should never both win an item.  
With this modification, \citeauthor{HR09} showed:
\begin{theorem}[\citet{HR09}]
\TheoremName{VCG2Approx}
Let $\OPT$ be the optimal revenue extractable in a $k$-items auction with $n$ unit-demand bidders whose values are independently drawn from regular, non-identical distributions. 
The expected revenue of the VCG auction with every bidder duplicated
is at least half of $\OPT$.
\end{theorem}

In fact, \citeauthor{HR09} showed a more general result for matroid settings.  As we discuss in the conclusion, we leave the question open whether our result for $k$-items auctions can be extended to general matroids.


\section{The Vickrey Auction with a Single Duplicate Bidder}
\label{sec:single}
\subsection{Warm-up: a loose bound for single-item auctions}
\label{sec:warm-up}
We illustrate in this section the main ideas behind our result on single-duplicate auctions and prove a $40$-approximation.
In \Section{single-item}, we optimize the parameters and get an approximation ratio better than~$10$.

Let $\{\exanteqi\}_{i\in[n]}$ be the optimal solution to the ex ante relaxation and let $\OPT$ denote
the optimal value of the ex ante relaxation.
In the case where $\vali(1/4) \geq \OPT / 2$ for some $i$,
posting a price of $\OPT / 2$ to bidder $i$ obtains a revenue of at least $\OPT / 8$.
Hence, Bulow-Klemperer's Theorem (\Theorem{BK}) implies that running a second-price auction with only two copies of bidder~$i$ already obtains a revenue of at least $\OPT / 8$;
introducing additional bidders only weakly increases that revenue and therefore in this case the SPA with bidder~$i$ duplicated obtains a revenue of at least $\OPT / 8$.

Henceforth, we may assume $\vali(1/4) < \OPT / 2$ for all $i$.
We show that in this case the SPA without duplicates in fact already obtains an $O(1)$-approximation of the optimal revenue (and introducing bidders never hurts the revenue).
Let $S$ be $\{i : \val_i(\exanteq_i) \geq \OPT / 2\}$.
Intuitively, bidders not in~$S$ have low prices in the optimal ex ante solution, and altogether contribute at most a constant fraction in the ex ante optimal.  Specifically,
\[
\sum_{i \notin S} \val_i(\exanteq_i) \exanteq_i \leq \frac{\OPT}{2} \sum_{i \notin S} \exanteq_i \leq \frac{\OPT}{2}
\]
where the first inequality is because $\val_i(\exanteq_i) \leq \OPT / 2$ for $i \notin S$ and the second inequality is
because $\sum_{i \notin S} \exanteq_i \leq 1$.
Hence,
\begin{equation}
    \EquationName{rev_bound}
    \sum_{i \in S} \Rev_i(\exanteq_i) = \sum_{i \in S} \val_i(\exanteq_i) \exanteq_i \geq \OPT / 2.
\end{equation}
For $i \in S$, define $q_i' = q_i(\OPT / 2)$ and notice that $\exanteq_i \leq q_i' < 1/4$.
We show that the \emph{sum} of $q_i'$'s is at least a constant, which in turn allows us to show that with constant probability, at least two bidders in~$S$ bid at least $\OPT /2$.

By concavity of the revenue curve (see \Claim{reg}), we have $\Rev_i(q_i') \geq \tfrac 3 4 \Rev_i(\exanteq_i)$.
Plugging this bound into \Equation{rev_bound}, we get that
\begin{equation}
    \EquationName{rev_bound2}
    \sum_{i \in S} \val_i(q_i') q_i' = \sum_{i \in S} \Rev_i(q_i') \geq \frac{3}{4} \sum_{i \in S} \Rev_i(q_i) \geq \frac{3 \OPT}{8}.
\end{equation}
On the other hand, since $\val_i(q_i') = \OPT / 2$ (by definition of $q_i'$), we have
\begin{equation}
    \EquationName{rev_bound3}
    \sum_{i \in S} \val_i(q_i') q_i' = \frac{\OPT}{2} \sum_{i \in S} q_i'.
\end{equation}
Combining \Equation{rev_bound2} and \Equation{rev_bound3} gives
\begin{equation}
    \EquationName{prob_bound}
    \sum_{i \in S} q_i' \geq \frac{3}{4}.
\end{equation}
We now show that with probability $\Omega(1)$, at least \emph{two} bidders bid at least $\OPT / 2$.
Indeed, let $X$ be the random variable that counts the number of bidders in $S$ that bid at least $\OPT / 2$.
Then $\prob{X = 0} = \prod_{i \in S} (1-q_i') \leq \exp(-\sum_{i \in S} q_i') \leq \exp(-3/4)$ and
\begin{align*}
    \prob{X = 1}
    & = \sum_{i \in S} q_i' \prod_{j \in S \setminus \{i\}} (1-q_j') \\
    & \leq \sum_{i \in S} q_i' \exp\left( -\sum_{j \in S \setminus \{i\}} q_j' \right) \\
    & \leq \sum_{i \in S} q_i' \exp\left( 1/4 - \sum_{j \in S} q_j' \right) && \text{(because $1/4-q_i' \geq 0$)} \\
    & = \exp(1/4) \cdot \left(\sum_{i \in S} q_i'\right) \cdot \exp\left(-\sum_{i \in S} q_i'\right)  \\
    & \leq \exp(-3/4),
\end{align*}
where the last inequality is because the function $f(x) = x \exp(-x)$ is maximized at $x = 1$.\footnote{
    This follows from standard calculus. Indeed, note that $f'(x) = (1-x) \exp(-x)$. Hence, $f(x)$ is non-decreasing on $(-\infty, 1]$
    and non-increasing on $[1,\infty)$. So $x = 1$ is the maximizer of $f$ and $f(1) = 1/e$.
}
To conclude, we have $\prob{X \geq 2} = 1 - \prob{X = 0} - \prob{X = 1} \geq 1 - 2\exp(-3/4) > 1/20$.
Hence, SPA extracts a revenue of at least $\OPT / 40$.

\subsection{A tighter bound for single-item auctions}
\label{sec:single-item}
We now optimize the parameters in the argument outlined in Section~\ref{sec:warm-up} and obtain main result for single-item auctions:
\begin{theorem}
\TheoremName{single-item}
In a single item auction with $n$ bidders whose values are drawn independently from (non-identical) regular distributions $F_1, \cdots, F_n$, there exists~$i$ such that the second price auction with the same $n$ bidders and an additional bidder $n+1$ whose value is drawn independently from~$F_i$ achieves at least $0.108$ fraction of the ex ante optimal revenue with the original $n$ bidders.
\end{theorem}

Recall from Section~\ref{sec:prelim} that $\exanteqi[1], \ldots, \exanteqi[n]$ are quantiles that solve the ex ante revenue maximization problem with the original $n$~bidders.
We use $\OPT = \sum_i \Rev_i(\exanteqi)$ to denote the ex ante optimal revenue, which upper bounds the optimal revenue (without duplicates).

Lemma~\ref{lem:single-technical} is the technical heart of the proof.  Using regularity, i.e.\@ the concavity of the revenue curves, it shows that either there is a bidder~$i$ who bids a high value with a constant probability,
or the \emph{sum} of each bidder's probability of bidding a high value must be large.  

\begin{lemma}
\LemmaName{single-technical}
Let $\alpha, \beta \in [0,1]$ be constants.
Consider a single item auction with $n$ bidders whose values are drawn independently from  regular distributions $F_1, \ldots, F_n$.
Let $\OPT$ be the ex ante optimal revenue.
Then exactly one of the following statements is true.
\begin{enumerate}
\item There exists $i \in [n]$ such that $\vali(\beta) \geq \alpha \cdot \OPT$.
\item For all $i$, $\val_i(\beta) < \alpha \cdot \OPT$ but $\sum_{i \in [n]} q_i(\alpha \cdot \OPT) \geq \frac{1-\alpha}{\alpha} \cdot (1 - \beta)$.
\end{enumerate}
\end{lemma}
\begin{remark}
    The proof will show that if the first condition does not hold in the lemma then the second condition holds which implies that at least one of the two
    statements in the lemma is true.
    However, the two statements are mutually exclusive so this also implies that \emph{exactly} one of the statements is true.
\end{remark}
\begin{proof}[Proof of \Lemma{single-technical}.]
    Let $\exanteq_1, \ldots, \exanteq_n$ be an optimal solution to the ex ante relaxation.
    In particular, $\sum_{i=1}^n \Rev_i(\exanteq_i) = \OPT$.
    Let $S = \{i : \val_i(\exanteq_i) \geq \alpha \cdot \OPT\}$.
    Since $\sum_{i \notin S} \exanteq_i \val_i(\exanteq_i) \leq \alpha \cdot \OPT$ (because $\sum_{i \notin S} \exanteq_i \leq 1$),
    we have $\sum_{i \in S} \exanteq_i \val_i(\exanteq_i) \geq (1-\alpha) \cdot \OPT$.

    Suppose that the first statement of the lemma does not hold, i.e.~$\val_i(\beta) < \alpha \cdot \OPT$ for all~$i$.  Recall that $q_i(\cdot)$ is a non-increasing function.  So for any $i \in S$, $\exanteq_i \leq q_i(\alpha \cdot \OPT) \leq \beta$ by definition.
    By the concavity of the revenue curve (see \Claim{reg}), we have $\Rev_i(q_i(\alpha \cdot \OPT)) \geq (1-\beta) \Rev_i(\exanteq_i)$ for all $i \in S$.  Therefore 
    \begin{align*}
	    \sum_{i \in S} \Rev_i(q_i(\alpha \cdot\OPT)) \geq (1 - \beta) \sum_{i \in S} \Rev_i(\exanteqi) \geq (1 - \beta) (1 - \alpha) \OPT.
    \end{align*}

    But $\Rev_i(q_i(\alpha\cdot\OPT))$ is just $(\alpha \cdot \OPT) q_i(\alpha \cdot \OPT)$.
    Rearranging, we have $\sum_{i \in S} q_i(\alpha \cdot \OPT) \geq \frac{1 - \alpha}{\alpha} \cdot (1 - \beta)$, as claimed.
\end{proof}

We are now ready to prove \Theorem{single-item}.
\begin{proof}[Proof of \Theorem{single-item}.]
	Let $\alpha, \beta \in [0,1]$ be constants such that $\frac {1 - \alpha}{\alpha} \cdot (1 - \beta) \geq 1$.  We will determine their values later.
The proof amounts to analyzing the two cases given in \Lemma{single-technical}.
\paragraph{Case 1: There exists $i \in [n]$ such that $\vali(\beta) \geq \alpha \cdot \OPT$.}
In this case, posting a price of $\vali(\beta)$ \emph{only} to bidder $i$ obtains a revenue of $\beta \vali(\beta) \geq (\alpha \beta) \cdot \OPT$.
Hence, we can apply the Bulow-Klemperer Theorem (\Theorem{BK}) to assert that duplicating bidder~$i$ suffices to get an $(\alpha \beta)$-fraction of the 
ex ante optimal revenue.
Finally, adding the remaining bidders back in can only increase the revenue so this obtains at least an $(\alpha\beta)$-fraction of the ex ante optimal revenue.

\paragraph{Case 2: For all $i$, $\vali(\beta) < \alpha \cdot \OPT$ but $\sum_{i \in [n]} q_i(\alpha \cdot \OPT) \geq \frac{1-\alpha}{\alpha} \cdot (1-\beta)$.}
To ease notation, let $q_i' \coloneqq q_i(\alpha \cdot \OPT)$.
Also, observe that the condition $\vali(\beta) < \alpha \cdot \OPT$ implies that $q_i' \leq \beta$.

In this case, let us duplicate any bidder in $\argmax_{i \in [n]} q_i'$, say $i^*$, and write $S' = [n] \cup \{i^*\}$.
Now, we will compute the probability that there are at least two bidders who bid at least $\alpha \cdot \OPT$.
The probability that no bidder bids above $\alpha \cdot \OPT$ is at most
\begin{equation}
  \EquationName{NoBidderHigh}
  \prod_{i \in [n]} (1 - q_i') \leq \exp\left( - \sum_{i \in [n]} q_i' \right) \leq \exp\left( -\frac{1-\alpha}{\alpha} \cdot (1 - \beta) \right).
\end{equation}

Since we duplicated a bidder in $\argmax_{i \in [n]} q_i'$, the probability that exactly one bidder bids at least $\alpha \cdot \OPT$ is
\begin{equation}
\EquationName{ExactlyOneBidderHighv2}
\begin{aligned}
  \sum_{i \in S'} q_i' \prod_{j \in S' \setminus \{i\}} (1-q_j')
  & \leq
  \sum_{i \in S'} q_i' \exp \left( - \sum_{j \in S' \setminus \{i\}} q_j' \right) \\
  & \leq \sum_{i \in S'} q_i' \exp\left( -\sum_{j \in [n]} q_j' \right)  \quad\text{(since $q_{i*}' \in \argmax_{j \in S} q_j'$)} \\
  & = \left( \sum_{i \in S'} q_i' \right) \cdot \exp\left( -\sum_{i \in [n]} q_i' \right) \\
  & \leq
  \left(\beta + \frac{1 - \alpha}{\alpha} \cdot (1 - \beta)\right) \cdot \exp\left( - \frac{1 - \alpha}{\alpha} \cdot (1 - \beta)\right)
  \quad \text{(since $q_{i^*}' \leq \beta$)},
\end{aligned}
\end{equation}
where the last inequality uses the assumption that $\frac{1 - \alpha}{\alpha} \cdot (1 - \beta) \geq 1$ 
and the function $x\exp(-x)$ is decreasing for $x \geq 1$.
Define
\[
  \eta =
  1 - \left(1 + \beta + \frac{1 - \alpha}{\alpha} \cdot (1 - \beta)\right) \cdot \exp\left( - \frac{1 - \alpha}{\alpha} \cdot (1 - \beta)\right),
\]
which is the probability that at least two bidders bid at least $\alpha \cdot \OPT$.
We have shown that one can always duplicate one bidder to obtain in the Vickrey auction an $\alpha \min\{\beta, \eta\}$-fraction of the optimal revenue.
Choosing $\alpha = 0.27$ and $\beta = 0.4$ and verifying $\frac {1 - \alpha}{\alpha} \cdot (1 - \beta)$ is indeed greater than~$1$, we complete the proof
of \Theorem{single-item}.
\end{proof}

\subsection{Choosing the distribution to duplicate}
\label{sec:robust}
\citeauthor{BK96}'s result and the extension by \citeauthor{HR09} are also seen as prototypical results in \emph{prior-independent} mechanism design.  These results guarantee performances of ``detail-free'' auctions, such as the Vickrey auction, as long as the underlying distribution satisfies some mild property.  As we pointed out in the Introduction, one may not duplicate a single bidder with complete ignorance on the value distributions, since arbitrary or uniformly random duplication easily fails.  We show in this section that the information requirement for choosing the duplication is resilient and minimal.  We give a few variants of the forms of information sufficient for our purpose. 

\begin{theorem}
	\label{thm:single-beta-exact}
	\label{thm:single-beta-noisy}
	\label{thm:single-duplicate}
Consider a single item auction with $n$ bidders whose values are drawn independently from (non-identical) regular distributions $F_1, \ldots, F_n$.
Let $\OPT$ be the ex ante optimal revenue for this setting.
\begin{enumerate}
\item 
	\label{part:single-beta-exact}
Suppose for some $\beta \in [0,1]$, the auctioneer has access to $\Rev_i(\beta)$ (or equivalently, $\val_i(\beta)$) for all~$i$.
Then, by duplicating any bidder~$i$ that maximizes $\Rev_i (\beta)$ 
and running the Vickrey auction on the $n+1$ bidders,
the auctioneer can extract a revenue of at least $c_1(\beta) \cdot \OPT$,
where $c_1(\beta)$ is a constant depending only on $\beta$.
In particular, $c_1(0.355) \geq 0.099$.

\item 
\label{part:single-beta-noisy}
More generally, suppose that the auctioneer has access to the following oracle for each bidder~$i$ for some $\beta \in [0,1/2], \eps \in [0, 1]$:  the oracle returns $\Rev_i(\beta_i')$ with the promise that
$|\beta_i' - \beta| \leq \eps \beta$.
Then, by duplicating any bidder~$i$ that maximizes $\Rev_i(\beta_i')$ 
and running the Vickrey auction on the $n+1$ bidders,
the auctioneer can extract a revenue of at least $c_2(\beta, \eps) \cdot \OPT$,
where $c_2(\beta, \eps)$ is a constant depending only on $\beta$ and $\eps$.
Moreover, $c_2(\beta, \eps) \geq (1-\eps)c_1(\beta)$ where $c_1(\beta)$ is the constant as from part~\ref{part:single-beta-exact}.

\item
\label{part:single-duplicate}
Suppose that the auctioneer can draw a sample $\samplei$, independently, from each distribution~$F_i$.
Then duplicating any bidder in $\argmax_i \samplei$ and running VCG on the resulting $n+1$ bidders yields $0.044$ fraction of the ex ante optimal revenue.
	\end{enumerate}
\end{theorem}

As we emphasized in the Introduction, revenue-wise, running a Vickrey auction with a duplicate may not be the best way to use the forms of prior information in Theorem~\ref{thm:single-duplicate}, but the Vickrey auction has the advantage of being simple and anonymous, and may be particularly useful when individual reserve prices cannot be placed.  We see Theorem~\ref{thm:single-item} as the main result of this work, and Theorem~\ref{thm:single-duplicate} as a natural consequence of our techniques.

\begin{proof}[Proof of \Theorem{single-beta-noisy} part~\ref{part:single-beta-noisy}]  We prove part~\ref{part:single-beta-noisy} directly, and part~\ref{part:single-beta-exact} is implied by the proof.

Let $\alpha \in [0,1]$ be any constant that satifies $(1-\alpha)(1-\beta)/\alpha \geq 1$.
As in the proof of \Theorem{single-item}, we analyze the two cases given by the conclusion of \Lemma{single-technical}.
\paragraph{Case 1: There exists $i \in [n]$ such that $\vali(\beta) \geq \alpha \cdot \OPT$.}
Let $i^* \in \argmax_{i \in [n]} \Rev_i(\beta)$.
Since we duplicate a bidder in $\argmax_{i \in [n]} \Rev_i(\beta_i')$, say $j$, we can apply the Bulow-Klemperer Theorem (\Theorem{BK})
to assert that second price auction with only bidder $j$ already obtains a revenue of at least $\Rev_j(\beta_j)$.
However, note that $\Rev_j(\beta_j') \geq \Rev_{i^*}(\beta_{i^*}') \geq (1-\eps) \Rev_{i^*}(\beta) \geq (1-\eps) \alpha \beta \OPT$,
where the second inequality uses \Claim{reg_delta2}.

\paragraph{Case 2: For all $i$, $\vali(\beta) < \alpha \cdot \OPT$ but $\sum_{i \in [n]} q_i(\alpha \cdot \OPT) \geq \frac{1-\alpha}{\alpha} \cdot (1-\beta)$.}
To ease notation, let $q_i' \coloneqq q_i(\alpha \cdot \OPT)$.
In this case, we will show that the probability that at least two bidders bid at least $\alpha \cdot \OPT$ is quite high even if we do not duplicate any bidder.
Indeed, the probability that no bidder bids at least $\alpha \cdot \OPT$ is
\[
  \prod_{i \in [n]} (1 - q_i') \leq \exp\left( - \sum_{i \in [n]} q_i' \right) \leq \exp\left( -\frac{1-\alpha}{\alpha} \cdot (1 - \beta) \right)
\]
and the probability that exactly one bidder bids at least $\alpha \cdot \OPT$ is
\begin{equation}
\begin{aligned}
  \sum_{i \in [n]} q_i' \prod_{j \in [n] \setminus \{i\}} (1-q_j')
  & \leq
  \sum_{i \in [n]} q_i' \exp \left( - \sum_{j \in [n] \setminus \{i\}} q_j' \right) \\
  & \leq \sum_{i \in [n]} q_i' \exp\left( \beta -\sum_{j \in [n]} q_j' \right)  \quad\text{(since $q_i' \leq \beta$ for all $i$)} \\
  & = \left( \sum_{i \in [n]} q_i' \right) \cdot \exp\left( \beta - \sum_{i \in [n]} q_i' \right) \\
  & \leq
  \left(\frac{1 - \alpha}{\alpha} \cdot (1 - \beta)\right) \exp(\beta) \cdot \exp\left( - \frac{1 - \alpha}{\alpha} \cdot (1 - \beta)\right).
\end{aligned}
\end{equation}
Here, we assume that $\left( \frac{1 - \alpha}{\alpha} \cdot (1-\beta) \right)\geq 1$ and that $x \cdot\exp \left (x \right)$ is non-increasing on $[1, \infty)$.  Hence, the probability that at least two bidders bid above $\alpha \cdot \OPT$ is at least
\[
  \eta(\alpha, \beta) =
  1 -
  \left(1 + \frac{1 - \alpha}{\alpha} \cdot (1 - \beta)\right) \exp(\beta) \cdot \exp\left( - \frac{1 - \alpha}{\alpha} \cdot (1 - \beta)\right).
\]
Hence, the revenue \emph{without} any extra bidders is at least $\alpha \eta(\alpha, \beta) \OPT$, so if we add any duplicate,
we are still guaranteed at least this quantity.

Combining the two cases implies that we obtain a revenue of at least $\alpha \cdot \min\{(1-\eps) \beta, \eta(\alpha, \beta)\} \OPT$.
In particular, one can take
\[
c(\beta, \eps) = \max_{\alpha \in [0,1] : (1-\alpha)(1-\beta)/\alpha \geq 1} \alpha \cdot \min\{(1-\eps) \beta, \eta(\alpha, \beta)\},
\]
completing the proof.
\end{proof}

\begin{proof}[Proof of \Theorem{single-duplicate}, part~\ref{part:single-duplicate}]
Let $\alpha, \beta, \gamma \in (0,1)$ be constants.
For technical reasons, we also need that $(1-\alpha)(1-\beta)/\alpha \geq 1$.
As in the proof of \Theorem{single-item}, we analyze the two cases given by the conclusion of \Lemma{single-technical}.
\paragraph{Case 1: There exists $i \in [n]$ such that $\vali(\beta) \geq \alpha \cdot \OPT$.}
Let $H = \{i \in [n] : \vali(\beta) \geq \alpha \cdot \OPT\}$ and $L = [n] \setminus H$.
We break things down into two smaller cases.
\paragraph{Case 1a: $\Prx{\text{$\vali \geq \alpha \cdot \OPT$ for some $i \in L$}} < \gamma$.}
In this case, with probability at least $\beta (1-\gamma)$, if we draw a single sample from each bidder then the
value of the highest sample will be in $H$.
Conditioned on this event, we duplicate a bidder in $H$ and obtain a revenue of at least $(\alpha \beta) \cdot \OPT$.
So in this case, we get at least $(\alpha \beta^2 (1-\gamma)) \cdot \OPT$.

\paragraph{Case 1b: $\Prx{\text{$\vali \geq \alpha \cdot \OPT$ for some $i \in L$}} \geq \gamma$.}
In this case, even if we do not duplicate any bidders, the probability that at least two bidders bid at least $\alpha \cdot \OPT$
is at least $\beta \gamma$.
Hence, the second price auction already extracts at least $(\alpha \beta \gamma) \cdot \OPT$.

\paragraph{Case 2: For all $i$, $\vali(\beta) < \alpha \cdot \OPT$ but $\sum_{i \in [n]} q_i(\alpha \cdot \OPT) \geq \frac{1-\alpha}{\alpha} \cdot (1-\beta)$.}
In this case, the argument is identical to that of part~\ref{part:single-beta-noisy} of \Theorem{single-beta-noisy}.
Let $\eta(\alpha, \beta)$ be as in the proof of part~\ref{part:single-beta-noisy}.

Putting all the cases together,
we see that this mechanism achieves a revenue of at least $\alpha \cdot \min \{\beta^2 (1 - \gamma), \beta \gamma, \eta(\alpha, \beta)\} \cdot \OPT$.
Choosing $\alpha = 0.26, \beta = 0.51, \gamma = 0.34$ yields a revenue of at least $0.0446 \OPT$.
\end{proof}

\section{The VCG Auction with $k$ Duplicates in $k$-Items Auctions}
\label{sec:k-item}
In this section we describe our results for $k$-items auctions with unit demand bidders.
The generalization of the Vickrey auction in this setting is the VCG auction, as we explained in Section~\ref{sec:prelim}.
Again we will show that, in order for the VCG auction to secure a constant fraction of the optimal auction's revenue, one needs to duplicate far fewer bidders than required by Theorem~\ref{thm:VCG2Approx}. However, one must think carefully about conditions that should be imposed on the duplication environment.  As we showed in the Introduction, the example with one bidder having value~$1$ and all other bidders having value~$0$ shows that even if one duplicates all bidders, the VCG auction still has revenue~$0$ whereas the optimal revenue is~$1$.  

The VCG auction therefore needs more power in this setting than simply duplicating bidders at most once.  We study two settings that grant the VCG auction different powers.  In the first setting, which we deem more natural and term as the ``free'' setting, the VCG auction may duplicate the same bidder more than once.
In the example above, it is easily seen that the VCG auction needs to duplicate the same bidder at least $k$ times to be a constant approximation.
Our main theorem in this section shows that duplicating $k$ bidders is sufficient.

\begin{theorem}
	\label{thm:kUnitFree}
In a $k$-item auction with $n$ unit demand bidders whose values are drawn independently from (non-identical)
regular distributions $F_1, \ldots, F_n$, there exists $i$ such that the VCG auction
with the same $n$ bidders and additional bidders $n+1, \ldots, n+k$ whose values are drawn independently
from $F_i$ achieves at least $0.009$ fraction of the ex ante optimal revenue with the original $n$ bidders.
\end{theorem}
We relegate the proof of Theorem~\ref{thm:kUnitFree} to Subsection~\ref{ss:kUnitFree}.
Although the proof follows a similar structure to the proof of \Theorem{single-item},
there are some subtleties that one needs to take into account.
First, it is not difficult to see that if there are a small number of bidders that are high bidders \emph{or}
the probability that many bidders bid above a certain threshold is large then one can extract a large fraction of the revenue.
The tricky part is to deal with the middle ground: when there is no bidder with high revenue \emph{and} it is not
often the case that many bidders bid above a certain threshold.
In this case, we show that the expected number of high bids is large from which the desired result will follow
from an application of a theorem due to \citet{Hoeffding56}.

The second setting was introduced by \citet{HR09}, which we term as ``constrained''.
Here the VCG auction can only duplicate each bidder at most once, but each original bidder and her duplicate can win at most one item.
With this constraint, in the example above, with the first bidder duplicated, the VCG auction recovers the optimal revenue.
The following example shows that $\Omega(k)$ duplicates are necessary, even with the additional constraint, for the VCG auction to recover a constant fraction of the optimal revenue: $k/3$ bidders have value~$1$, and the rest of the bidders have value~$0$.
The optimal revenue here is $k/3$, and the VCG auction's revenue is equal to the number of duplicates it recruits.
We show that this is the worst case: $k$ duplicates always suffices for the VCG auction to extract a constant fraction of the optimal revenue.

\begin{theorem}
\TheoremName{kUnitApprox2}
Let $\OPT$ be the ex ante optimal revenue in a $k$-items auction with $n$ unit-demand bidders with values drawn independently from regular (non-identical) distributions. There exist $k$ bidders such that the VCG auction with a duplicate of each of these $k$ bidders, where at most one of an original bidder and her duplicate can win an item, extracts  revenue that is at least $0.1 \OPT$.
\end{theorem}
We relegate the proof of \Theorem{kUnitApprox2} to \Subsection{kUnitApprox2}.
Let us also remark that in this setting, we can also obtain an analagous result to \Theorem{single-beta-exact}.
We relegate the details to Appendix~\ref{sec:app-k-item}.

The proofs of \Theorem{kUnitFree} and \Theorem{kUnitApprox2} both make use of the following analog of \Lemma{single-technical}
for $k$-item auctions.
\begin{lemma}
    \LemmaName{kItemsCasesNew}
    Let $\beta, \gamma, \delta \in (0,1)$ be contants that satisfy $\frac{(1-\gamma)(1-\beta) - \delta}{\gamma} \geq \frac{3}{2}$.
    In a $k$-items auction with $n$ unit demand bidders whose values are drawn independently from regular distributions $F_1, \ldots, F_n$
    at least one of the following three statements must be true:
    \begin{enumerate}
        \item There exists $q_1, \ldots, q_n \in [\beta,1)$ and a set $H \subseteq [n]$ with $|H| \leq k$
            such that $\val_i(\beta) \geq \frac{\gamma}{k} \OPT$ for all $i \in H$
            and $\sum_{i \in H} \Rev_i(q_i) \geq \delta \cdot \OPT$.
        \item There exists a set $H \subseteq [n]$ with $|H| = k$ such that $\val_i(\beta) \geq \frac{\gamma}{k} \OPT$.
        \item With probability $1/2$, at least $k+1$ bidders bid at least $\frac{\gamma}{k} \OPT$.
    \end{enumerate}
\end{lemma}
\begin{proof}
    Assume that the first two statements do not hold.
    Let $\{ \exanteq_i \}_{i=1}^n$ be an optimal solution to the ex ante relaxation
    and $S = \{i \in [n] \ : \ \val_i(\exanteq_i) \geq \frac{\gamma}{k} \OPT \}$.
    Then $\sum_{i \notin S} \Rev_i(\exanteq_i) \leq \gamma \OPT$ because $\sum_{i \notin S} \exanteq_i \leq k$.
    Hence, $\sum_{i \in S} \Rev_i(\exanteq_i) \geq (1-\gamma) \cdot \OPT$.

    For $i \in S$, define $q_i'$ as follows.
    If $\exanteq_i > \beta$ then set $q_i' = \exanteq_i$ but note that $\val_i(\beta) \geq \val_i(q_i') \geq \frac{\gamma}{k} \OPT$ since
    $\val_i(\cdot)$ is decreasing in $[0,1]$.
    Otherwise, set $q_i' = \min\{ \beta, 1 - F_i(\frac{\gamma}{k} \OPT) \}$.
    Note that if $q_i' < \beta$ then $\val_i(q_i') = \frac{\gamma}{k} \OPT$.

    If $q_i' \leq \beta$ then \Claim{reg} implies that $\Rev_i(q_i') \geq (1-\beta) \Rev_i(\exanteq_i)$
    and if $q_i' > \beta$ then $q_i' = \exanteq_i$ so $\Rev_i(q_i') = \Rev_i(\exanteq_i)$.
    Hence $\sum_{i \in S} \Rev_i(q_i') \geq (1-\gamma)(1-\beta) \cdot \OPT$.
    Let $S' = \{ i \in S \ : \ \val_i(\beta) \geq \frac{\gamma}{k} \OPT \}.$ If $i \notin S'$ we have $v_i(\beta) < \frac{\gamma}{k} \OPT$, which means that $q_i' \leq \beta$ (if $q_i' > \beta$ then $q_i' = \exanteq_i$). If $q_i' = \beta$, then $v_i(q_i') < \frac{\gamma}{k} \OPT$ because $i \notin S'$. Otherwise, $\val_i(q_i') = \val_i \left(1 - F_i(\frac{\gamma}{k} \OPT \right) = \frac{\gamma}{k} OPT$. Hence, for $i \not in S'$, we have $\val_i(q_i') \leq \frac{\gamma}{k}\OPT$. 
    Since the second item in the claim does not hold, we have $|S'| < k$.
    Since the first item in the claim does not hold, we have
    \begin{align*}
        \sum_{i \in S \setminus S'} \Rev_i(q_i')
        = \sum_{i \in S} \Rev_i(q_i') - \sum_{i \in S'} \Rev_i(q_i')
        \geq \left( (1 - \gamma)(1-\beta) - \delta \right) \OPT.
    \end{align*}
    As we observed above, $\val_i(q_i') \leq \frac{\gamma}{k} \OPT$ for all $i \in S \setminus S'$.
    Plugging this into the above inequality with the fact that $\Rev_i(q_i') = q_i' \val_i(q_i')$ gives
    \begin{align*}
        \sum_{i \in S \setminus S'} q_i' \geq k \cdot \frac{(1-\gamma)(1-\beta) - \delta}{\gamma} \geq \frac{3k}{2},
    \end{align*}
    where the last inequality is an assumption made in the claim.

    Let $X$ be the number of bidders in $S \setminus S'$ that bid at least $\frac{\gamma}{k} \OPT$.
    Then $\expect{X} \geq \sum_{i \in S \setminus S'} q_i' \geq \frac{3k}{2} \geq k+1$ since $k \geq 2$.
    By \Theorem{Hoeffding} and \Fact{median}, we have that the probability that at least $k+1$ bidders in $S \setminus S'$ bid at least $\frac{\gamma}{k} \OPT$
    is at least $1/2$.
\end{proof}

\subsection{Proof of Theorem~\ref{thm:kUnitFree}}
\label{ss:kUnitFree}
\begin{proof}[Proof of Theorem~\ref{thm:kUnitFree}.]
Let $\beta, \gamma, \delta \in (0,1)$ be constants to be chosen later.
We will consider the three cases in the conclusion of \Lemma{kItemsCasesNew}.

\paragraph{Case 1: There exists $q_1, \ldots, q_n \in [\beta,1)$ and a set $H \subseteq [n]$ with $|H| \leq k$ such
    that $\sum_{i \in H} \Rev_i(q_i) \geq \delta \cdot \OPT$.}
By shuffling the indices, we may assume without loss of generality that $\Rev_1(q_1) \geq \ldots \geq \Rev_n(q_n)$.
Let $r \leq k$ be such that $\Rev_i(q_i) \geq \frac{\delta \OPT}{2k}$ for all $i \leq r$ and $\Rev_i(q_i) < \frac{\delta \OPT}{2k}$ for
$i \in \{r+1, \ldots, k\}$.
Then $\sum_{i=1}^r \Rev_i(q_i) \geq \frac{\delta \OPT}{2}$ because $\sum_{i = r+1}^k \Rev_i(q_i) < k \frac{\delta \OPT}{2k} = \frac{\delta \OPT}{2}$.

Define $R_i = \frac{\Rev_i(q_i)}{\OPT}$.
We will consider two subcases based on the value of $R_1$.

\paragraph{Case 1a: $R_1 \geq \frac{\delta}{4}$.}
First, note that $\val_1(1-1/k) \geq \frac{\delta}{4k} \OPT$.
To see this, if $q_1 > 1-1/k$ then $\val_1(1-1/k) \geq \val_1(q_1) \geq \Rev_1(q_1) \geq \frac{\delta}{4} \OPT$.
On the other hand, if $q_1 \leq 1 - 1/k$ then \Claim{reg} implies that $\Rev_1(1-1/k) \geq \Rev_1(q_1) / k \geq \frac{\delta}{4k} \OPT$
and hence $\val_1(1-1/k) \geq \frac{\delta}{4k}\OPT $.

Now suppose we obtain $k$ duplicates of bidder $1$.
The probability that any one of them bids above $\frac{\delta}{4k} \OPT$ is at least $1 - 1/k$ so the probability that
all $k+1$ of them do is at least $(1-1/k)^{k+1} \geq 1/8$ since $k \geq 2$.
Hence, VCG extracts a revenue of at least $\frac{\delta \OPT}{32}$.

\paragraph{Case 1b: $R_1 \leq \frac{\delta}{4}$.}
We begin with the following claim which gives a lower bound on the probability
that a bidder bids above a threshold of $\frac{\delta \OPT}{8k}$.
\begin{claim}
    For all $i \in [r]$, we have $\val_i\left( 1 - \frac{\delta}{8k R_i} \right) \geq \frac{\delta \OPT}{8k}$.
\end{claim}
\begin{proof}
    Suppose first that $q_i > 1 - \frac{\delta}{8kR_i}$.
    Then since $v_i(\cdot)$ is non-increasing on $[0,1]$, we have
    \[
        v_i\left( 1 - \frac{\delta}{8kR_i} \right) \geq v_i(q_i) \geq \Rev_i(q_i) \geq \frac{\delta \OPT}{2k}
    \]
    where the last inequality is by definition of $r$.

    On the other hand, suppose that $q_i \leq 1 - \frac{\delta}{8kR_i}$.
    Then we can apply \Claim{reg} to get that
    \begin{align*}
        \Rev_i\left( 1 - \frac{\delta}{8kR_i} \right)
        \geq \frac{\delta}{8kR_i} \Rev_i(q_i)
        = \frac{\delta \OPT}{8k}.
    \end{align*}
    Consequently, $\val_i\left( 1 - \frac{\delta}{8k R_i} \right) \geq \frac{\delta \OPT}{8k}$.
\end{proof}

Consider obtaining $k$ duplicates of bidder $1$.
Let $X$ be the number of bidders amongst the original $r$ bidders and the $k$ duplicates that bid at least $\frac{\delta \OPT}{8k}$.
Then
\begin{align*}
    \expect{X}
    \geq k \left( 1 - \frac{\delta}{8kR_1} \right) + \sum_{i=1}^r \left( 1 - \frac{\delta}{8k R_i} \right)
    \geq k - \frac{\delta}{8R_1} + \frac{3r}{4},
\end{align*}
where the last inequality is because $R_i \geq \frac{\delta}{2k}$ for $1 \leq i \leq r$ so $1 - \frac{\delta}{8k R_i} \geq \frac{3}{4}$.
It now remains to lower bound $r$.
\begin{claim}
    $r \geq \frac{\delta}{2R_1}$.
\end{claim}
\begin{proof}
    Indeed, we have
    \[
        r
        = \sum_{i=1}^r R_i \frac{1}{R_i}
        \geq \frac{1}{R_1} \sum_{i=1}^r R_i
        \geq \frac{\delta}{2 R_1}.
    \]
    The first inequality is because $R_i \leq R_1$ for all $i$.
    The second inequality is because $\sum_{i=1}^r \Rev_i(q_i) \geq \frac{\delta \OPT}{2}$
    is equivalent to $\sum_{i=1}^r R_i \geq \frac{\delta}{2}$.
\end{proof}
Hence, we have
\begin{equation*}
    \expect{X} \geq k - \frac{\delta}{8R_1} + \frac{3r}{4} \geq k - \frac{\delta}{8R_1} + \frac{3\delta}{8R_1} = k + \frac{\delta}{4R_1}.
\end{equation*}
Since $R_1 \leq \delta / 4$, we have that $\expect{X} \geq k + 1$.

Now observe that $X$ is a sum of independent Bernoulli random variables, albeit with non-identical means.
Combining \Theorem{Hoeffding} and \Fact{median}, we conclude that at least $k + 1$ bidders bid at least $\frac{\delta \OPT}{8k}$
with probability at least $1/2$.
Hence, in this case VCG obtains a revenue of at least $\frac{\delta \OPT}{16}$.

\paragraph{Case 2: There exists a set $H \subseteq [n]$ with $|H| = k$ such that $\val_i(\beta) \geq \frac{\gamma}{k} \OPT$.}
Note that for $i \in H$, we have $\Rev_i(\beta) \geq \frac{\beta \gamma}{k} \OPT$.
We claim that this implies $\val_i(3/4) \geq \frac{\beta \gamma \OPT}{3k}$.
Indeed, if $\beta \geq 3/4$ then $\val_i(3/4) \geq \val_i(\beta) = \Rev_i(\beta) / \beta \geq \Rev_i(\beta) \geq \frac{\beta \gamma}{k} \OPT$.
On the other hand, if $\beta \leq 3/4$ then $\Rev_i(3/4) \geq \frac{\beta \gamma}{4k} \OPT$ by \Claim{reg}.
Hence, $\val_i(3/4) = \frac{\Rev_i(3/4)}{3/4} \geq \frac{\beta \gamma \OPT}{3k}$.

Now consider obtaining $k$ duplicates of any bidder in $H$.
Then the average number of bidders from $H$ and the duplicates that bid at least $\frac{\beta \gamma\OPT}{3k}$ is at least
$\frac{3}{4} \cdot 2k = \frac{3k}{2} \geq k+1$.
By \Theorem{Hoeffding} and \Fact{median} again, this implies that VCG obtains a revenue of at least $\frac{\beta \gamma}{6} \OPT$.

\paragraph{Case 3: With probability at least $1/2$ there exists at least $k + 1$ bidders who bid at least $\frac{\gamma}{k} \OPT$.}
The assumption implies that VCG without duplicates already obtains a revenue of at least $\frac{\gamma}{2} \OPT$;
duplicating bidders will never hurt the revenue.

Let $\cA = \left\{ (\beta, \gamma, \delta) \ : \ \frac{(1-\gamma)(1-\beta) - \delta}{\gamma} \geq \frac{3}{2} \right\}$.
Putting all the cases together, we get the revenue of VCG with $k$ duplicates is at least
\begin{equation*}
    \OPT \cdot \max_{(\beta, \gamma, \delta) \in \cA} \min\left\{ \frac{\delta}{32}, \frac{\beta \gamma}{6}, \frac{\gamma}{2} \right\}.
\end{equation*}
By setting $\beta = 0.377, \delta = 0.3, \gamma = 0.15$, we see that this is at least $0.009 \cdot \OPT$.
\end{proof}
\begin{remark}
    One can also modify the above proof so that the choice of the bidder is independent of $k$, albeit with a decay in the approximation factor.

    Indeed, notice that in case 1 of the above proof, one can duplicate any bidder in
    \[
        \argmax_{i \in [n]} \max_{q_i \in [\beta, 1]} \Rev_i(q_i)
    \]
    and that this choice is independent of $k$.
    By reordering, suppose that bidder $1$ is in the argmax of the previous expression.

    Note that it may \emph{not} be the case that $\val_1(\beta) \geq \frac{\gamma}{k} \OPT$
    so we may not be able to duplicate bidder 1 in case 2 of the above proof.
    However, if $q_1^* \in \argmax_{q_1 \in [\beta, 1]} \Rev_1(q_1)$ then $\Rev_1(q_1^*) \geq \max_{i \in [n]}\Rev_i(\beta) \geq \frac{\beta\gamma}{k} \OPT$.
    Hence,
    \[
        \val_1(\beta) \geq \val_1(q_1^*) \geq q_1^* \val_1(q_1^*) = \Rev_1(q_1^*) \geq \frac{\beta \gamma}{k} \OPT.
    \]
    So one can repeat the argument in case 2 of the above proof with a set $H \subseteq [n]$ and $|H| = k$ such that
    $1 \in H$ and $\val_i(\beta) \geq \frac{\beta \gamma}{k} \OPT$ for all $i \in H$.
    Since we can duplicate any bidder in $H$, we can duplicate bidder $1$ as well which would give $\frac{\beta^2 \gamma}{6}$-approximation.

    Case 3 of the above proof remains unchanged.

    Hence, the approximation ratio is at least $\min\left\{ \frac{\delta}{32}, \frac{\beta^2 \gamma}{6}, \frac{\gamma}{2} \right\}$
    (provided that $\beta, \gamma, \delta$ satisfy the conditions of \Theorem{kUnitFree}).
    Choosing $\beta = 0.4, \delta = 0.2, \gamma = 0.19$ gives an approximation ratio of $> 0.005$.
\end{remark}

\subsection{Proof of \Theorem{kUnitApprox2}}
\SubsectionName{kUnitApprox2}
\begin{proof}[Proof of \Theorem{kUnitApprox2}]
We will consider the three cases in the conclusion of \Lemma{kItemsCasesNew}.

\paragraph{Case 1: There exists $q_1, \ldots, q_n \in (0,1)$ and a set $H \subseteq [n]$ with $|H| \leq k$ such
    that $\sum_{i \in H} \Rev_i(q_i) \geq \delta \cdot \OPT$.}
Suppose we duplicate every bidder in $H$ and run $\VCG$ with just these bidders and their duplicates,
while ensuring that a bidder and her duplicate do not both win.
In this setting, the VCG auction simply runs a single-item second-price auction for each pair.
Observe that if the auctioneer had posted a price of $\Rev_i(q_i) / q_i = \vali(q)$ for bidder $i$ then the revenue
extracted from bidder $i$ would have been at least $\Rev_i(q_i)$.
Applying \Theorem{BK}, it follows that the single-item auction with bidder $i$ and her duplicate yields a revenue of at least $\Rev_i(q_i)$.
Hence, the revenue of VCG is at least $\sum_{i \in H} \Rev_i(q_i) \geq \delta \cdot \OPT$.

\paragraph{Case 2: There exists a set $H \subseteq [n]$ with $|H| = k$ such that $\val_i(\beta) \geq \frac{\gamma}{k} \OPT$.}
This case is very similar to the previous case.
Again, we duplicate every bidder in $H$ and note that VCG with only the bidders in $H$ and their duplicates is exactly
a single-item second-price auction for each pair.
Since $\Rev_i(\beta) \geq \frac{\beta \gamma}{k} \OPT$ for $i \in H$, we can apply \Theorem{BK} to assert that we extract at least
$\frac{\beta \gamma}{k} \OPT$ revenue from bidder $i$ or her duplicate.
Hence, the total revenue is at least $\beta \gamma \OPT$.

\paragraph{Case 3: With probability at least $1/2$ there exists at least $k + 1$ bidders who bid at least $\frac{\gamma}{k} \OPT$.}
In this case, the $\VCG$ mechanism without duplicates already gets a $\gamma/2$ fraction of the ex ante optimal revenue with the original $n$ bidders.
Now observe that VCG with the original $n$ bidders and their duplicates (with the constraint that a bidder and her duplicate cannot simultaneously win) earns at least as much revenue of VCG for the original $n$ bidders.
Hence, VCG in the duplicate setting also achieves at least $\gamma / 2$ fraction of the ex ante optimal revenue.

Choosing $\beta = 0.5, \delta = 0.1, \gamma = 0.2$ and checking that $\frac{(1-\gamma)(1-\beta) - \delta}{\gamma} = \frac{3}{2}$
gives the desired bound of $0.1 \cdot \OPT$.
\end{proof}

\section{Tighter Analysis for the Second Price Auction with $n$ Single Duplicated Bidders}
\label{sec:2approx}
In this section, we strengthen Theorem~\ref{thm:HR} \citep{HR09} in another direction.

\begin{theorem}
	\TheoremName{2approx}
	Consider an auction with $2$ bidders with $\vali[1], \vali[2]$ drawn independently from regular distributions $F_1$ and~$F_2$.  Let $\OPT$ be the ex ante optimal revenue.  The expected revenue of the Vickrey auction with $4$~bidders, with $\vali[1], \vali[3]$ drawn from~$F_1$ and $\vali[2], \vali[4]$ from~$F_2$, all independently, is at least $\frac{3}{4} \OPT$.
\end{theorem}

The proof of \Theorem{2approx} appears in \Subsection{2approx}.
This ratio matches a lower bound given by \citet{HR09} (see Example~\ref{ex:lb-HR}), and therefore is tight for $n = 2$.
We first show that the Vickrey auction with $n$ duplicates generates at least as much revenue as the \emph{lookahead auction} without the duplicates.  We will use the following auction as a middle step, whose revenue obviously lower bounds the Vickrey auction with duplicates.

\begin{definition}[The Second Price Auction with Late Duplicate (SPALD)]
Solicit bids from bidders $1$ to~$n$, and let $i^*$ be the highest bidder among them.  Then run the second price auction for bidders $1, 2, \ldots, n$ and $n + i^*$.
\end{definition}

\begin{lemma}
\LemmaName{spa-lookahead}
Let $\LA$ denote the revenue of the lookahead auction with $n$ bidders, where each bidder~$i$'s value is independently drawn from a regular distribution~$F_i$.  The Vickrey auction with $n$ duplicates
extracts a revenue at least as much as~$\LA$.
\end{lemma}

The proof of Lemma~\ref{lem:spa-lookahead} is in \Subsection{spald_proof}.
It draws ingredients from both \citet{R01}'s analysis of the lookahead auction and \citet{DRY15}'s use of the revenue curve.\footnote{In fact,
    the case $n = 1$ is the Bulow-Klemperer Theorem for a single bidder for which \citet{DRY15} gave an elegant proof.
    The proof of \Lemma{spa-lookahead} can be seen as generalizing the proof of \citet{DRY15} to more than a single bidder.}
Conditioning on $\valsmi$, the profile of values except bidder~$i$'s,
it expresses SPALD's revenue from bidder~$i$ and her duplicate in terms of an area under a curve closely related to bidder~$i$'s revenue curve.
The comparison with the lookahead auction's revenue then follows from the concavity of the curve.

\begin{remark}
    \Theorem{lookahead} and \Lemma{spa-lookahead} can be used to give a short and transparent proof of \Theorem{HR}.
    Indeed, \Theorem{lookahead} asserts that the revenue of the lookahead auction is a 2-approximation to the optimal revenue and \Lemma{spa-lookahead}
    implies that the Vickrey auction with $n$ duplicates extracts at least as much revenue as the lookahead auction.
    Hence, the Vickrey auction with $n$ duplicates is a 2-approximation to the optimal revenue with the original $n$ bidders.
\end{remark}

To prove Theorem~\ref{thm:2approx},
we again switch our benchmark to the stronger ex ante optimal revenue.
It was observed in \citet{AHNPY15} that, given an optimal solution $\exanteqi[1], \ldots, \exanteqi[n]$ to the ex ante revenue maximization problem,
the solution remains the same if one were to replace the revenue curves by triangles that each peaks at $\exanteqi$ with revenue~$\Rev_i(\exanteqi)$.
On the other hand, for the Vickrey auctions,
the triangle revenue curves are also the worst-case distributions (Lemma~\ref{lem:spa-flat}).
This allows us to focus on triangle revenue curves and then analytically minimize the ratio between the SPALD's revenue and the ex ante optimal revenue.

We recall the example from \citet{HR09} showing that this ratio is tight for $n = 2$: 

\begin{example}
	\label{ex:lb-HR}
	Let $F_1$ be the point mass on $\vali[1] = 1$, and let $F_2$ be the ``slanted'' equal revenue distribution $F_2(\val) = 1 - \frac{1}{\val + 1}$. If we offer a price of $H \gg 1$ to bidder $2$ and, if not taken, offer a price of $1$ to bidder $1$ then we obtain a revenue of $H/(H+1) + 1 - (1/(H+1)) \approx 2$.
On the other hand, in the second price auction with duplicates of each distribution, if we denote by $w$ the second highest value, then the revenue is 
\begin{align*}
	\Ex{w} = \int_0^{\infty} \Prx{w \geq v} \:\dd v = 1 + \int_1^{\infty} \left( \frac 1 {v + 1} \right)^2 \: \dd v = \frac 3 2.
\end{align*}
\end{example}

\begin{remark}
	\label{rem:lb-HR}
	Note that in Example~\ref{ex:lb-HR}, the two revenue curves are two triangle revenue curves with $R_1 = R_2 = 1$, $\exanteqi[1] = 1$ and $\exanteqi[2] = 0$.\footnote{The point mass can be seen to have a triangle revenue curve peaking at $q = 1$ by being approximated by a uniform distribution on $[1, 1+\eps]$ with arbitrarily small~$\eps$.}
\end{remark}

\begin{remark}
	\label{rem:lb-single}
    Example~\ref{ex:lb-HR} also provides a lower bound on the approximation of the revenue of the Vickrey auction
    with a \emph{single} duplicate which is \emph{greater} than with both duplicates.
    Indeed, if one duplicates only the first bidder, then the resulting Vickrey auction has revenue~$1$.
    On the other hand, duplicating the second bidder extracts revenue strictly less than $1/2$.
    To see this, note that if both of the duplicates bid strictly less than $1$ than the revenue extracted in this case is strictly less than $1$
    (whereas if we duplicate both bidders then the revenue would still be $1$ in this case).
    With more detailed calculations, we show in Appendix~\ref{sec:app-2approx} that the revenue is $\ln 4$ and hence, the approximation ratio is
    $\approx 1.44$.
    This is the largest gap we know for the Vickrey auction with a single duplicate.
    \Xcomment{
	Example~\ref{ex:lb-HR} also provides a lower bound on the Vickrey auction's revenue with a \emph{single} duplicate: by replicating the first bidder,
    the resulting Vickrey auction has revenue~$1$,
    half of the optimal revenue without duplicates; and by replicating the second bidder, the Vickrey auction has revenue $\log 4$.
    So the approximation ratio is $1.44$.
    This is also the largest gap we know for the Vickrey auction with a single duplicate.
    The detailed calculations can be found in Appendix~\ref{sec:app-2approx}.
    }
\end{remark}

We believe that the Vickrey auction with $n$ duplicates in general give a $4/3$-approximation to the optimal revenue,
although the following example shows that for $n > 2$ it does \emph{not} $4/3$-approximate the \emph{ex ante} optimal revenue.  
	
\begin{example}
    \label{ex:n=3}
    Consider three distributions, the first having a triangle revenue curve peaking at (0, 1) (i.e., having cdf $F_1(\val) = 1 - \frac 1 {\val + 1}$),
    and the other two having triangle revenue curve peaking at $(\frac 1 2, \frac 1 2)$.
    The ex ante optimal revenue is~$2$,
    but we show in Appendix~\ref{sec:app-2approx} that the Vickrey auction with $3$ duplicates has a revenue strictly less than~$1.5$.
\end{example}

\subsection{Proof of Lemma~\ref{lem:spa-lookahead}.}
\SubsectionName{spald_proof}
\begin{proof}[Proof of Lemma~\ref{lem:spa-lookahead}.]

	We say bidder~$n+i$ is bidder~$i$'s duplicate.  We will show that SPALD has a revenue no less than~$\LA$.  Let us condition on the realization of $\vali[1], \ldots, \vali[i^* - 1], \vali[i^* + 1], \ldots, \vali[n]$ (which henceforth we denote as $\vals_{[n] - i^*}$), and denote by $\sechigh$ the highest among them.  
We show that, conditioning on this, the lookahead auction (without duplicates) obtains no more revenue from bidder~$i^*$ than SPALD does from bidder~$i^*$ \emph{and} her duplicate, bidder $n + i^*$.

We first analyze the posted price faced by $i^*$ in the lookahead auction.  Recall that the lookahead auction posts the revenue-optimal price for $i^*$ conditioning on $\vali[i^*] \geq \sechigh$.
Conditioning on $\vali[i^*] \geq \sechigh$, the revenue by posting any price $p \geq \sechigh$ to bidder~$i^*$ is $p(1 - F_{i^*}(p)) / (1 - F_{i^*}(\sechigh))$, which is simply the unconditioned revenue at that price $p (1 - F_{i^*}(p))$ amplified by a constant factor $1 / (1 - F_{i^*}(\sechigh))$.   Therefore the shape of the revenue curve conditioned on $\vali[i^*] \geq \sechigh$ is simply bidder $i^*$'s revenue curve to the left of $q(\sechigh)$ and amplified by a constant factor.
Let $r_i^*$ denote the monopoly reserve for bidder $i$.
Since the original revenue curve is concave, it is nondecreasing to the left of~$q_{i^*}(\monres_{i^*})$, the quantile of the monopoly reserve price, and non-increasing to the right.  Therefore, if $\sechigh$ is higher than $\monres_{i^*}$, 
the lookahead auction will simply post $\sechigh$ as the take-it-or-leave-it price for bidder~$i^*$; otherwise, the lookahead auction uses $\monres_{i^*}$ as the price.\footnote{This analysis essentially shows that the lookahead auction in our setting is the second price auction with monopoly reserve prices.}

If $\sechigh \geq \monres_{i^*}$, the lookahead auction extracts $\sechigh$ from bidder~$i^*$ when $\vali[i^*] \geq \sechigh$; but whenever $\vali[i^*] \geq \sechigh$, SPALD extracts a revenue at least $\sechigh$ from bidders $i^*$ and $n + i^*$.

The case where $\sechigh < \monres_{i^*}$ needs more calculation.  Let $R_{i^*}$ be the monopoly revenue extractable from bidder~$i^*$ (i.e., $\monres_{i^*}(1 - F_{i^*}(\monres_{i^*})$).  The lookahead auction, by posting the monopoly reserve price $\monres_{i^*}$, extracts revenue $R_{i^*}$ from bidder~$i^*$.  
We show that SPALD gets at least as much from bidder~$i^*$ and her duplicate.  

Recall that the duplicate bidder's value $\vali[n + i^*]$ is drawn from the same distribution $F_{i^*}$ as bidder~$i^*$.
With probability $F_{i^*}(\sechigh)$, $\vali[n + i^*]$ is no more than~$\sechigh$, and SPALD extracts a revenue of $\Rev_{i^*}(q_{i^*}(\sechigh))$ from bidder~$i^*$.  When $\vali[n+ i^*] \geq \sechigh$, bidder~$i^*$ faces in SPALD a take-it-or-leave-it price of $\vali[n + i^*]$, and the expected revenue from bidder~$i^*$is precisely the height of $\Rev_{i^*}$ at the quantile of~$\vali[n+i^*]$.
Since $\vali[n+ i^*]$ is drawn from the same distribution, its quantile in $F_{i^*}$, namely, $q_{i^*}(\vali[n + i^*]) = 1 - F_{i^*}(\vali[n + i^*])$, is uniformly distributed on $[0, 1]$.  Therefore the expected revenue when $\vali[n + i^*] \geq \sechigh$ is 
simply the area under bidder~$i^*$'s revenue curve between $0$ and~$q_{i^*}(\sechigh)$.
(See Figure~\ref{fig:rev-curve} for an illustration.\footnote{This conversion from the expected revenue to an area under the revenue curve is an elegant technique from \citet{DRY15}.}) 
The expected revenue from the duplicate bidder~$n + i^*$ is the same area: when $\vali[i^*]$ is at least~$\sechigh$, bidder~$n + i^*$ faces a take-it-or-leave-it price of $\vali[i^*]$, whose quantile in $F_{n + i^*}$ is uniform between $0$ and~$q_{n+i^*}(\sechigh)$, and therefore the expected revenue from bidder~$n+i^*$ is exactly the area under the revenue curve~$\Rev_{i^*}$ between quantiles $0$ and~$q_{i^*}(\sechigh)$.

Summing up the parts, the revenue of SPALD from bidders $i^*$ and~$n + i^*$ when $\sechigh \leq \monres_{i^*}$ is twice the area under the revenue curve $\Rev_{i^*}$ between $0$ and~$q_{i^*}(\sechigh)$,
plus the area of a rectangle with height $\Rev_{i^*}(q_{i^*}(\sechigh))$ and width $1 - q_{i^*}(\sechigh)$.  
If one replaces the part of the revenue curve $\Rev_{i^*}$ between $q_{i^*}(\monres_{i^*})$ and~$1$ by a straight line, the resulting curve is still concave, with the highest value at~$R_{i^*}$.  
Twice the area under this new curve is precisely SPALD's revenue from $i^*$ and~$n + i^*$.  By the concavity of the curve, this is still at least $R_{i^*}$.  This completes the proof.
\begin{figure}
\centering
\begin{tikzpicture}
\fill [fill = gray!30] (0, 0) -- (3, 0) -- (3, 1.166) -- (1, 3.5) -- cycle;
\draw [thick, ->] (0, 0) -- (5, 0) node[anchor = north west] {$q$};
\draw [thick, ->] (0, 0) -- (0, 4) node[anchor = south east] {$\Rev$};
\draw [thick] (0, 0) node[below]{$1$}  -- (1, 3.5) -- (4, 0) node[below] {$0$};
\draw [dashed] (1, 0) node[below]{$q(\monres)$} -- (1, 3.5);
\draw [dashed] (3, 0) node[below]{$q(\sechigh)$} -- (3, 1.166);
\end{tikzpicture}
\caption{Illustration of a concave revenue curve.  In this figure $\sechigh < \monres$, and the shaded area is the expected revenue from this bidder when both she and her duplicate bidder bid above $\sechigh$.}
\label{fig:rev-curve}
\end{figure}
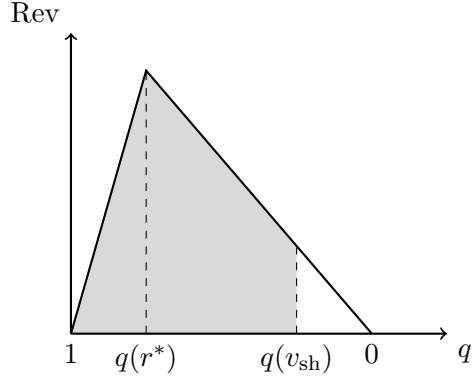
\end{proof}

\subsection{Proof of \Theorem{2approx}}
\SubsectionName{2approx}
\begin{proof}[Proof of Theorem~\ref{thm:2approx}]
	Given any two regular distributions with revenue curves $\Rev_1$ and~$\Rev_2$, let $\exanteqi[1], \exanteqi[2]$ be the solution to the ex ante revenue maximization problem, and let $R_1$ be $\Rev_1(\exanteqi[1])$ and $R_2$ be~$\Rev_2(\exanteqi[2])$.  Then the ex ante optimal revenue is $R_1 + R_2$.

	If we replace the two distributions so that bidder~$i$'s revenue curve is a triangle with the three vertices at $(0, 0), (\exanteqi, R_i)$ and $(1, 0)$, then by Lemma~\ref{lem:spa-flat}, the expected revenue of the second price auction with two duplicates is smaller whereas, by Lemma~\ref{lem:ex-ante-flat}, the ex ante optimal revenue remains the same.  We therefore only need to show that the second price auction with duplicates extracts at least $0.75$ fraction of the ex ante optimal revenue in this setting with triangular revenue curves.

	By Lemma~\ref{lem:spa-lookahead}, it suffices to analyze the revenue of the lookahead auction (without duplicates) in this setting.  Without loss of generality, suppose $\frac{R_1}{\exanteqi[1]} \geq \frac{R_2}{\exanteqi[2]}$.  By the same analysis as in the proof of Lemma~\ref{lem:spa-lookahead}, bidder~$1$ always faces a take-it-or-leave-it price of~$R_1 / \exanteqi[1]$, by which the auction extracts a revenue of~$R_1$ from her.  As for bidder~$2$, when $\vali[1] < R_2 / \exanteqi[2]$, she is faced with a price at $R_2 / \exanteqi[2]$; otherwise she faces a price of $\vali[1] > \exanteqi[2]$, which she can afford with probability~$0$.  Therefore the expected revenue extracted from bidder~$2$ is $R_2$ times $\Prx{\vali[1] \leq R_2 / \exanteqi[2]}$.  This probability can be easily calculated from the geometry of the revenue curves: let this probability be~$\unknown$, then 
	\begin{align*}
		\frac{\unknown R_1}{1 - \exanteqi[1]} = \frac{(1 - \unknown) R_2}{\exanteqi[2]} \quad \Rightarrow \unknown = \frac{\alpha (1 - \exanteqi[1])}{\exanteqi[2] + (1 - \exanteqi[1]) \alpha},
	\end{align*}
	where $\alpha$ denotes the ratio $\frac{R_2}{R_1}$.  The ratio between  the revenue of the lookahead auction and that of the ex ante optimal revenue is therefore
	\begin{align*}
		\frac{R_1 + R_2 \unknown}{R_1 + R_2} = \frac{1 + \frac{\alpha^2 (1 - \exanteqi[1])}{\exanteqi[2] + (1 - \exanteqi[1])\alpha}}{1 + \alpha}.
	\end{align*}
	We are interested in the minimum value of this ratio.  It is evident that, everything else fixed, the ratio decreases with $\exanteqi[2]$.  Therefore, given the constraint $\exanteqi[1] + \exanteqi[2] \leq 1$, we may take $\exanteqi[2] = 1 - \exanteqi[1]$.  This simplifies the ratio to
	\begin{align*}
		\frac{1 + \frac{\alpha^2 \exanteqi[2]}{\exanteqi[2] + \alpha \exanteqi[2]}}{1 + \alpha} = \frac{1 + \alpha + \alpha^2}{(1 +\alpha)^2}.
	\end{align*}
	This ratio is minimized at $\alpha = 1$, and evaluates to~$\frac 3 4$.
\end{proof}

\section{Conclusion and Open Questions}
\label{sec:conclusion}
In this work we closed the gap in the number of duplicates needed for the VCG auction to yield a constant fraction of the optimal revenue.  For both the single-item and $k$-item auctions, we showed tight bounds: instead of duplicating all bidders as required by \citet{HR09}, only one (or $k$, respectively) duplicate is needed.  These match the numbers in the i.i.d.\@ settings in Bulow and Klemperer's original result, and show that the only loss, when one removes symmetry from the i.i.d.\@ setting, is a constant fraction of revenue, without need to grow the number of duplicates (although one needs to be more careful in choosing a duplicate).

We leave several interesting questions for future study.  
We have focused on obtaining a constant approximation, and 
our technique is unlikely to yield the tightest approximation achieveable by a single duplicate.  In fact, in the single-item auction, the worst revenue loss with a \emph{single} duplicate we know of is only $0.31$ fraction of the optimal revenue (see Remark~\ref{rem:lb-single} in Section~\ref{sec:2approx}).  We believe the approximation ratio is closer to $2$ than to~$10$.  

We extended our results to $k$-items auctions in two settings.  In the ``free'' setting, a bidder can be duplicated multiple times and all duplicates participate as all other bidders.  This setting is most natural for a $k$-items auction, but cannot be extended to settings with more complex feasibility constraints.  The ``constrained'' setting allows each bidder to be duplicated at most once, but allows at most one of a bidder and her duplicate to win.  This can be generalized to domains such as matroid settings. \citet{HR09} showed their $2$-approximation for $n$-duplicates in any matroid setting.  It would be interesting to know whether $k$ duplicates always suffice for an auction with a rank~$k$ matroid feasibility costraint.

We made progress on the ten-year-old open problem of the approximation ratio of the Vickrey auction with $n$ duplicates.  We showed that $\tfrac 4 3$ is the correct bound for $n = 2$ but as we remarked in Section~\ref{sec:2approx}, new techniques would be needed to tighten the analysis for $n > 2$.

\newpage

\bibliographystyle{apalike}
\bibliography{bibs}

\begin{thebibliography}{}

\bibitem[Alaei et~al., 2015]{AHNPY15}
Alaei, S., Hartline, J.~D., Niazadeh, R., Pountourakis, E., and Yuan, Y.
  (2015).
\newblock Optimal auctions vs. anonymous pricing.
\newblock In {\em {IEEE} 56th Annual Symposium on Foundations of Computer
  Science, {FOCS} 2015, Berkeley, CA, USA, 17-20 October, 2015}, pages
  1446--1463.

\bibitem[Allouah and Besbes, 2018]{AB18}
Allouah, A. and Besbes, O. (2018).
\newblock Prior-independent optimal auctions.
\newblock In {\em Proceedings of the 2018 {ACM} Conference on Economics and
  Computation, Ithaca, NY, USA, June 18-22, 2018}, page 503.

\bibitem[Bulow and Klemperer, 1996]{BK96}
Bulow, J. and Klemperer, P. (1996).
\newblock Auctions versus negotiations.
\newblock {\em The American Economic Review}, 86(1):180--194.

\bibitem[Bulow and Roberts, 1989]{BR89}
Bulow, J. and Roberts, J. (1989).
\newblock The simple economics of optimal auctions.
\newblock {\em Journal of Political Economy}, 97(5):1060--90.

\bibitem[Chawla et~al., 2013]{CHMS13}
Chawla, S., Hartline, J.~D., Malec, D.~L., and Sivan, B. (2013).
\newblock Prior-independent mechanisms for scheduling.
\newblock In {\em Symposium on Theory of Computing Conference, STOC'13, Palo
  Alto, CA, USA, June 1-4, 2013}, pages 51--60.

\bibitem[Dhangwatnotai et~al., 2015]{DRY15}
Dhangwatnotai, P., Roughgarden, T., and Yan, Q. (2015).
\newblock Revenue maximization with a single sample.
\newblock {\em Games and Economic Behavior}, 91:318--333.

\bibitem[Eden et~al., 2017]{EFFTW17}
Eden, A., Feldman, M., Friedler, O., Talgam{-}Cohen, I., and Weinberg, S.~M.
  (2017).
\newblock The competition complexity of auctions: {A} bulow-klemperer result
  for multi-dimensional bidders.
\newblock In {\em Proceedings of the 2017 {ACM} Conference on Economics and
  Computation, {EC} '17, Cambridge, MA, USA, June 26-30, 2017}, page 343.

\bibitem[Feldman et~al., 2018]{FFR18}
Feldman, M., Friedler, O., and Rubinstein, A. (2018).
\newblock 99{\%} revenue via enhanced competition.
\newblock In {\em Proceedings of the 2018 {ACM} Conference on Economics and
  Computation, Ithaca, NY, USA, June 18-22, 2018}, pages 443--460.

\bibitem[Fu et~al., 2015]{FILS15}
Fu, H., Immorlica, N., Lucier, B., and Strack, P. (2015).
\newblock Randomization beats second price as a prior-independent auction.
\newblock In {\em Proceedings of the Sixteenth {ACM} Conference on Economics
  and Computation, {EC} '15, Portland, OR, USA, June 15-19, 2015}, page 323.

\bibitem[Hartline, 2017]{MDnA}
Hartline, J. (2017).
\newblock {\em Mechanism Design and Approximation}.

\bibitem[Hartline and Roughgarden, 2009]{HR09}
Hartline, J.~D. and Roughgarden, T. (2009).
\newblock Simple versus optimal mechanisms.
\newblock In {\em Proceedings 10th {ACM} Conference on Electronic Commerce
  (EC-2009), Stanford, California, USA, July 6--10, 2009}, pages 225--234.

\bibitem[Hoeffding, 1956]{Hoeffding56}
Hoeffding, W. (1956).
\newblock On the distribution of the number of successes in independent trials.
\newblock {\em Ann. Math. Statist.}, 27(3):713--721.

\bibitem[Kaas and Buhrman, 1980]{KB80}
Kaas, R. and Buhrman, J.~M. (1980).
\newblock Mean, median and mode in binomial distributions.
\newblock {\em Statistica Neerlandica}, 34(1):13--18.

\bibitem[Myerson, 1981]{M81}
Myerson, R. (1981).
\newblock Optimal auction design.
\newblock {\em Mathematics of Operations Research}, 6(1):pp. 58--73.

\bibitem[Ronen, 2001]{R01}
Ronen, A. (2001).
\newblock On approximating optimal auctions.
\newblock In {\em ACM Conference on Electronic Commerce}, pages 11--17.

\bibitem[Roughgarden et~al., 2012]{RTY12}
Roughgarden, T., Talgam{-}Cohen, I., and Yan, Q. (2012).
\newblock Supply-limiting mechanisms.
\newblock In {\em {ACM} Conference on Electronic Commerce, {EC} '12, Valencia,
  Spain, June 4-8, 2012}, pages 844--861.

\bibitem[Sivan and Syrgkanis, 2013]{SS13}
Sivan, B. and Syrgkanis, V. (2013).
\newblock Vickrey auctions for irregular distributions.
\newblock In {\em Web and Internet Economics - 9th International Conference,
  {WINE} 2013, Cambridge, MA, USA, December 11-14, 2013, Proceedings}, pages
  422--435.

\bibitem[Vickrey, 1961]{Vic61}
Vickrey, W. (1961).
\newblock Counterspeculation, auctions and competitive sealed tenders.
\newblock {\em Journal of Finance}, pages 8--37.

\end{thebibliography}
\newpage

\appendix

\section{Technical results}
\label{sec:technical}
\begin{claim}
\ClaimName{reg}
Suppose that $0 \leq q \leq q' \leq \beta \leq 1$.
Then for any regular distribution, we have $\Rev(q') \geq (1 - \beta) \Rev(q)$.
\end{claim}
\begin{proof}
By the concavity of the revenue curves, we have
\[
    \Rev(q') \geq \frac{1-q'}{1-q} \Rev(q) + \frac{q'-q}{1-q} \Rev(1) \geq (1-\beta) \Rev(q),
\]
since $\Rev(1) \geq 0$.
\end{proof}


\begin{claim}
\ClaimName{reg_delta2}
Let $q \in [0,1/2]$ and $q' \in [0,1]$ be such that $|q - q'| \leq \eps q$ for some $\eps \in [0,1]$.
Then $(1-\eps) \Rev(q) \leq \Rev(q') \leq \frac{1}{1-\eps} \Rev(q)$.
\end{claim}
\begin{proof}
First suppose that $q' > q$.
By concavity of the revenue curve, we have
\begin{align*}
\Rev(q') & \geq \frac{1 - q'}{1-q} \Rev(q) + \frac{q - q'}{1-q} \Rev(1) \\
& \geq \frac{1-q'}{1-q} \Rev(q) \\
& \geq \frac{1-q+q-q'}{1-q} \Rev(q) \\
& = \left( 1 - \frac{q' - q}{1-q} \right) \Rev(q).
\end{align*}
Since $q \leq 1/2$, we have $q \leq 1 - q$.
So
\begin{align*}
    \Rev(q')
    & \geq \left( 1 - \frac{q' - q}{1-q} \right) \Rev(q) \\
    & \geq \left( 1 - \frac{\eps q}{1-q} \right) \Rev(q) \\
    & \geq \left( 1 - \frac{\eps (1-q)}{1-q} \right) \Rev(q) \\
    & = (1-\eps) \Rev(q).
\end{align*}
This proves the lower bound on $\Rev(q')$.
To get the upper bound, we use concavity again and note that
\begin{align*}
    \Rev(q)
    & \geq \frac{q}{q'} \Rev(q') + \frac{q'-q}{q'} \Rev(0) \\
    & \geq \frac{q}{q'} \Rev(q') \\
    & \geq \frac{q' - \eps q}{q'} \Rev(q') \\
    & = (1 - \eps q / q') \Rev(q') \\
    & \geq (1-\eps) \Rev(q'),
\end{align*}
where the last inequality is because $q < q'$.
So $\Rev(q') \leq \frac{1}{1-\eps} \Rev(q')$.

The case of $q' < q$ follows from calculations which are similar to those above.
Indeed
\begin{align*}
\Rev(q') & \geq \frac{q'}{q} \Rev(q) + \frac{q-q'}{q} \Rev(0) \\
& \geq \frac{q'}{q} \Rev(q) \\
& \geq \frac{(1-\eps)q}{q} \Rev(q) \\
& = (1-\eps)\Rev(q)
\end{align*}
and
\[
\Rev(q)
\geq \left( 1 - \frac{q - q'}{1-q'} \Rev(q') \right).
\]
Since $q' < q \leq 1/2$, we have $q \leq 1 - q \leq 1 - q'$.
So $q - q' \leq \eps q \leq \eps(1-q')$ and this shows that $\Rev(q) \geq (1-\eps) \Rev(q')$ as desired.
\end{proof}

Let $\Ber(p)$ denote the Bernoulli distribution with mean $p$ and $\Binom(n,p)$ be the binomial distribution.
\begin{theorem}[Theorem~4 in \cite{Hoeffding56}]
\TheoremName{Hoeffding}
Let $X_1, \ldots, X_n$ be random variables where $X_i \sim \Ber(p_i)$.
Let $S = \sum_{i=1}^n X_i$ and set $p = n^{-1} \Ex S$.
Then for all $c \leq np - 1$, $\Pr[S \leq c] \leq \Pr[\Binom(n, p) \leq c]$.
\end{theorem}
\begin{fact}[Corollary~1 in \cite{KB80}]
\FactName{median}
The median of $\Binom(n,p)$ is at least $\lfloor np \rfloor$.
\end{fact}
Let $X_1, \ldots, X_n$ be random variables where $X_i \sim \Ber(p_i)$, $S = \sum_{i=1}^n X_i$,a nd $p = n^{-1} \Ex S$.
Since $np - 1 \leq \lfloor np \rfloor$, \Theorem{Hoeffding} and \Fact{median} imply that
\[
    \prob{S \leq np-1} \leq \prob{S < \lfloor np \rfloor} \leq 1/2.
\]
Hence, $\prob{S \geq np} \geq 1/2$.




\section{Additional Results for $k$-item Auctions}
\label{sec:app-k-item}
\label{sec:theoremKUnitApprox2}
\label{sec:k-items-cases-proof}
\label{sec:KUnitRobustness}
\label{sec:compare-rev}

\Theorem{kUnitApprox2} only guarantees the existence of at most $k$ bidders which can be duplicated so that $\VCG$ in the duplication
environment approximates the ex-ante optimal revenue.
The next result shows that under mild assumptions and (noisy) knowledge about the revenue curve,
the auctioneer can determine which bidders ought to be duplicated.

\begin{theorem}
\TheoremName{k-item-beta-noisy}
Consider a $k$-item auction with $n$ bidders whose values are drawn independently from (non-identical) regular distributions $F_1, \ldots , F_n$.
Suppose that the auctioneer has access to the following oracle for each bidder $i$ for some $\beta \in (0,1/2]$ and $\eps \in (0,1)$:
the oracle returns $\Rev_i(\beta_i')$ with the promise that $\abs{\beta - \beta_i} \leq \eps \beta$.
Let $S$ be the $k$ highest bidders in terms of $\Rev_i(\beta_i')$.
Then duplicating every bidder in $S$ and running the $\VCG$ mechanism
(restricting that a bidder and her duplicate cannot both win) on the $n+k$
bidders, the auctioneer can extract a revenue of at least $c(\beta, \epsilon) \cdot \OPT$ where
$c(\beta, \epsilon)$ is a constant depending only on $\beta$ and $\epsilon$.
Moreover, $c(0.5, \eps) \geq (1-\eps)^3 \cdot 0.1$.
\end{theorem}

\begin{proof}
Let $\gamma, \delta$ satisfy $\frac{1}{\gamma}((1-\beta)(1-\gamma) - \delta) \geq \frac{3}{2}$.
By shuffling the indices, let us assume that $\Rev_1(\beta) \geq \ldots \geq \Rev_n(\beta)$.

Again, we consider the following three cases given by \Lemma{kItemsCasesNew}.
However, we will have to be a little careful with how we set up the cases (note that the cases are not necessarily a partition).

\paragraph{Case 1: $\val_k(\beta) < (1-\eps)^2 \frac{\gamma}{k}$ but there exists $q_1, \ldots, q_n \in [\beta, 1)$
    and a set $H \subseteq [n]$ with $|H| < k$ such that $\val_i(\beta) \geq \frac{\gamma}{k}$ for all $i \in H$
    and $\sum_{i \in H} \Rev_i(q_i) \geq \delta \cdot \OPT$.}
The assumptions in the case imply that we can assume $H = \{ i \in S \ : \ \val_i(\beta) \geq \frac{\gamma \OPT}{k} \}$.
Let $S$ be the top $k$ highest bidders in terms of $\Rev_i(\beta_i')$.
We claim that $H$ is a subset of $S$.
\begin{claim}
    \ClaimName{HsubsetS}
    It holds that $H \subseteq S$.
\end{claim}
\begin{proof}
    Suppose $i \in H$.
    Then $\val_i(\beta) \geq \frac{\gamma}{k} \OPT$ so $\Rev_i(\beta) \geq \frac{\gamma \beta}{k} \OPT$.
    By \Claim{reg_delta2}, this implies that $\Rev_i(\beta_i') \geq (1-\eps) \frac{\gamma \beta}{k} \OPT$.

    On the other hand, suppose that $i \geq k$.
    Then $\val_i(\beta) < (1-\eps)^2\frac{\gamma}{k} \OPT$ so $\Rev_i(\beta) < (1-\eps)^2\frac{\gamma \beta}{k} \OPT$.
    Hence, applying \Claim{reg_delta2} again, we have that $\Rev_i(\beta_i') < (1-\eps) \frac{\gamma \beta}{k} \OPT$.

    The two previous paragraphs imply that if $i \in H$ then $\Rev_i(\beta_i') \geq \Rev_j(\beta_j')$ for all $j \geq k$.
    Since $S$ contains the top $k$ bidders in terms of $\Rev_i(\beta_i')$ this implies that $S$ contains $H$ completing the proof.
\end{proof}

Since we duplicate every bidder in $H$, \Theorem{BK} implies that a lower bound on the revenue of VCG on the resulting environment
is $\sum_{i \in H} \Rev_i(\beta) \geq \delta \cdot \OPT$.

\paragraph{Case 2: There exists a set $H \subseteq [n]$ with $|H| \geq k$ such that $\val_i(\beta) \geq (1-\eps)^2\frac{\gamma}{k} \OPT$.}
In this case, we have $\Rev_i(\beta) \geq (1-\eps)^2\frac{\beta \gamma}{k} \OPT$ for $1 \leq i \leq k$.
Let $H$ be the $k$ highest bidders in terms of $\Rev_i(\beta_i')$.
Then
\begin{align*}
    \sum_{i \in H} \Rev_i(\beta_i')
    & \geq \sum_{i=1}^k \Rev_i(\beta_i') \\
    & \geq \sum_{i=1}^k (1-\eps) \Rev_i(\beta) \quad \text{(by \Claim{reg_delta2})} \\
    & \geq \sum_{i=1}^k (1-\eps)^3 \frac{\beta \gamma}{k} \OPT \\
    & \geq (1-\eps)^3 \beta \gamma \OPT.
\end{align*}
Hence, duplicating the bidders in $H$ obtains and running VCG extracts a revenue of at least $(1-\eps)^3 \beta \gamma \OPT$.

\paragraph{Case 3: With probability at least $1/2$ there exists at least $k + 1$ bidders who bid at least $\frac{\gamma}{k} \OPT$.}
This case is straightforward since without duplicating, VCG already extracts a revenue of $\frac{\gamma}{2} \OPT$.

Combining the three cases we see that we can obtain a revenue of at least $\min \{ (1-\epsilon) \beta \delta, (1 - \epsilon) \beta \gamma, \gamma(1/2 )   \} \OPT$. In particular, one can take
\[
c(\beta, \eps) ~=~ \max_{\substack{\gamma, \delta \in [0,1] : \\ \frac{1}{\gamma}((1- \beta)(1-\gamma) - \delta) \geq 3/2} }
\min \left\{ \delta , (1-\eps)^3 \gamma \beta, \gamma / 2 \right\}.
\]
Finally, observe that if one chooses $\beta = 0.5, \delta = 0.1, \gamma = 0.2$ then one obtains that $c(0.5, \eps) \geq (1-\eps)^3 \cdot 0.1$.
\end{proof}


\section{Missing proofs from Section~\ref{sec:2approx}}
\label{sec:app-2approx}

\begin{lemma}
	\LemmaName{rev-dominate}
	Let $R$ and $\tilde R$ be the revenue curves of two value distributions $F$ and~$\tilde F$.  $F$ stochastically dominates~$\tilde F$ if and only if $R$ pointwise dominates~$\tilde R$ in the sense that for any $q \in [0, 1]$, $R(q) \geq \tilde R(q)$.
\end{lemma}

\begin{proof}
	Suppose $F$ stochastically dominates~$\tilde F$, for any quantile~$q$, $F^{-1}(1-q) \geq \tilde F^{-1}(1-q)$, and so $R(q) = q F^{-1}(1-q) \geq q \tilde F^{-1}(1-q) = \tilde R(q)$.   Conversely, if $R$ pointwise dominates~$\tilde R$, for any value~$\val$, let the quantile of~$\val$ in~$F$ be~$q$, and that in~$\tilde F$ be~$\tilde q$. Then $\Prx[w \sim F]{w \geq \val} = \frac{R(q)}{ \val} \geq \frac{\tilde R(q)} {\val} = \Prx[w \sim \tilde F]{w \geq \val}$, which shows that $F$ stochastically dominates~$\tilde F$.
\end{proof}

\begin{lemma}
	\LemmaName{spa-flat}
	In a second price auction where all bidders' values are drawn independently, if we replace bidder~$i$ by another bidder~$i'$ whose revenue curve is dominated by that of bidder~$i$'s, the new auction extracts no more expected revenue than the original.
\end{lemma}

\begin{proof}
	The revenue of the second price auction is the expected value of the second highest bid.  By Lemma~\ref{lem:rev-dominate}, the operation in the lemma statement replaces one value distribution by another stochastically dominated by the former.  It is immediate that this decreases the expected value of any order statistics, including the revenue of the second price auction.
\end{proof}
%
\begin{lemma}
	\LemmaName{ex-ante-flat}
	In a single item auction setting with values independently drawn from regular distributions whose revenue curves are $\Rev_1, \cdots, \Rev_n$, suppose the optimal solution to the ex ante revenue maximization problem are the quantiles $\exanteqi[1], \exanteqi[2], \ldots, \exanteqi[n]$, these remain an optimal solution (to the problem of ex ante revenue maximization) when each bidder~$i$'s value distribution is replaced by one whose revenue curve is a triangle whose three vertices are $(0, 0), (1, 0)$ and~$(\exanteqi, \Rev_i(\exanteqi))$.
\end{lemma}
This lemma was proved and used in \citet{AHNPY15}.
For completeness we give a proof.

\begin{proof}
	Note that $\exanteqi[1], \ldots, \exanteqi[n]$ constitute a feasible solution to the ex ante revenue in the setting with triangle revenue curves.  Therefore the ex ante optimal revenue in the new setting is no less than before.  However, it cannot be more either.  This is because each revenue curve in the new setting is pointwise dominated by the original one, by the concavity of the latter, so any feasible solution gives weakly less revenue in the new setting.  This shows the optimality of $\exanteqi[1], \ldots, \exanteqi[n]$.
\end{proof}

\begin{proof}[Calculation of Remark~\ref{rem:lb-single}]
    If we duplicate the deterministic bidder then the revenue is $1$.

    On the other, suppose that we duplicate the bidder whose cdf is $F_2(v) = 1 - \frac{1}{v+1}$.
    Let us analyze the revenue of the Vickrey auction in this case.
    Let $v_1$ be the value of the deterministic bidder (so $v_1 = 1$ with probability $1$)
    and $v_2, v_3 \sim F_2$.

    If $v_2, v_3 \leq 1$ (which occurs with probability $1/4$) then the Vickrey auction extracts $\max\{v_2, v_3\}$.
    So the average revenue, conditioned on $v_2, v_3 \leq 1$ is $\int_0^1 1 - \left( 2 - \frac{2}{t+1} \right)^2 \: \dd t = 4 \ln(4) - 5$.
    Hence, this case contributes $\ln(4) - 5/4$ to the revenue.

    If $v_2 \leq 1 \leq v_3$ or $v_3 \leq 1 \leq v_2$ (each occurring with probability $1/4$) then the revenue is $1$.
    These two cases contribute $1/2$ to the revenue.

    Finally, if $v_2, v_3 \geq 1$ (which occurs with probability $1/4$) then the revenue is $1 + \int_1^\infty \frac{4}{(t+1)^2} \: \dd t = 1 + 2$.
    So this case contributes $3/4$ to the revenue.

    Hence, the average revenue is $\ln(4)$.

    The revenue of the optimal auction with the original bidders is $\approx 2$ so the approximation ratio is $\frac{2}{\ln 4} \approx 1.44$.
\end{proof}

\begin{proof}[Calculation of Example~\ref{ex:n=3}]
	Note that the highest value in the supports of the latter two distributions is~$1$.  Let $u$ be the highest value among the bidders with values drawn from the latter two distributions.  Conditioning on any~$u$, the revenue of the Vickrey auction with duplicates is at most $u + \int_u^{\infty} \frac 1 {(v+1)^2} \: \dd v = u + \frac 1 {u + 1}$.  For $u \geq 0$ this is increasing in~$u$, and is $1.5$ when $u$ is~$1$.  However, with constant probability $u$ is strictly smaller than, say, $0.9$, and therefore in expectation the Vickrey auction with duplicates extracts a revenue strictly smaller than~$1.5$.
\end{proof}

\end{document}